\documentclass[a4paper,USenglish,cleveref,autoref]{lipics-v2021}

\nolinenumbers % pour la version finale

\pdfoutput=1 % uncomment to ensure pdflatex processing (mandatatory
\hideLIPIcs % uncomment to remove references to LIPIcs series (logo,
\usepackage[T1]{fontenc}
\usepackage{graphicx}
\usepackage[utf8]{inputenc}
\usepackage{amssymb,amsmath,xspace}
\usepackage{fixmath} % \boldmath{}
\usepackage[dvipsnames,table]{xcolor} % pour \color, \textcolor, \cellcolor (avant tikz)

\def\pathfig{figures}
\graphicspath{{\pathfig/}}
\definecolor{DarkGreen}{rgb}{0.1,0.7,0.1}
\definecolor{MyBlue}{RGB}{15,60,210}
\usepackage{latexcolors} % pour d'autres couleurs, cf. https://latexcolor.com/ 

\usepackage{tikz} 
\usetikzlibrary{arrows.meta}
\usetikzlibrary{decorations.pathreplacing,calligraphy}
\usetikzlibrary{fit,calc,positioning}
\usetikzlibrary{graphs} % pour \graph
\usetikzlibrary{calc} % pour $...$ dans les coordonnées
\tikzstyle{vertex}=[draw=black,fill=white,circle,inner sep=1.4pt]
\tikzstyle{accolade}=[decorate,decoration={brace,amplitude=5pt,raise=1.5ex}]

\usepackage{listofitems}
\def\listofdef#1#2#3{%
    \begingroup
    \readlist*\L{#3}%
    \foreachitem\i\in\L{%
    \expandafter\xdef\csname#1\i\endcsname{\noexpand\csname#2\endcsname{\i}}
    }%
    \endgroup
}

\listofdef{c}{mathcal}{A,C,F,G}
\listofdef{s}{mathscr}{C,F,K,L,S,T,U,V,W}
\listofdef{h}{widehat}{C,F,G,H,P,Q,S,T}

\definecolor{MyBlue}{RGB}{15,60,210}
\definecolor{DarkGreen}{rgb}{0.1,0.7,0.1}

\long\def\Cyril#1{\textcolor{MyBlue}{[Cyril: {#1}]}}
\long\def\Edgar#1{\textcolor{red}{[Edgar: {#1}]}}
\long\def\Francois#1{\textcolor{violet}{[Fran\c{c}ois: {#1}]}}
\long\def\Status#1{\textcolor{violet}{\fbox{\textbf{Status}}~[{#1}]}}

\long\def\Cyril#1{}
\long\def\Edgar#1{}
\long\def\Francois#1{}
\long\def\Status#1{}

\NewDocumentCommand\set{sm}{\IfBooleanTF#1{\{{#2}\}}{\left\{{#2}\right\}}}
\NewDocumentCommand\ceil{sm}{\IfBooleanTF#1{\lceil{#2}\rceil}{\left\lceil{#2}\right\rceil}}
\NewDocumentCommand\floor{sm}{\IfBooleanTF#1{\lfloor{#2}\rfloor}{\left\lfloor{#2}\right\rfloor}}
\NewDocumentCommand\pare{sm}{\IfBooleanTF#1{({#2})}{\left({#2}\right)}}
\NewDocumentCommand\range{smm}{\IfBooleanTF#1{\set*{{#2},\dots,{#3}}}{\set{{#2},\dots,{#3}}}}
\NewDocumentCommand\card{sm}{\IfBooleanTF#1{|{#2}|}{\left|{#2}\right|}}
\NewDocumentCommand\cardV{sm}{\IfBooleanTF#1{|V({#2})|}{\card{V\!\pare{#2}}}}

\def\leq{\leqslant}\def\le{\leq}
\def\geq{\geqslant}\def\ge{\geq}

\newcommand{\cupdot}{\mathbin{\mathaccent\cdot\cup}} % ne marche pas directement avec \bigcup
\DeclareMathOperator{\dist}{dist}

\newenvironment{myfigure}%
  {\begin{figure}[htbp!]\centering}%
  {\end{figure}}

\def\problem#1#2{\expandafter\def\csname #1\endcsname{\textsf{#2}\xspace}}
\problem{IUG}{Isometric-Universal Graph}
\problem{IUGG}{Isometric-Universal Graph for Two Graphs}
\problem{ASS}{Assignment}
\problem{MATCH}{3D-Matching}
\problem{HP}{Hamiltonian Path}
\problem{KP}{Clique}

\newenvironment{myquote}%
  {\begin{quote}\it\begin{tabular}{||p{.85\textwidth}}}
  {\end{tabular}\end{quote}}

\newtheorem*{fact}{Fact}
\def\myparagraph#1{\subparagraph*{\boldmath{#1}}} % pour LIPIcs

\title{Isometric-Universal Graphs for Trees}

\author{Edgar Baucher}%
{LaBRI, University of Bordeaux, CNRS, France}%
{edgar.baucher@labri.fr}%
{https://orcid.org/0009-0008-6053-6428}{}

\author{Fran\c{c}ois Dross}%
{LaBRI, University of Bordeaux, CNRS, France}%
{francois.dross@u-bordeaux.fr}%
{https://orcid.org/0000-0002-0535-9640}{}

\author{Cyril Gavoille}%
{LaBRI, University of Bordeaux, CNRS, France}%
{gavoille@labri.fr}%
{https://orcid.org/0000-0003-3671-8607}{}

\Copyright{Edgar Baucher, Fran\c{c}ois Dross, and Cyril Gavoille}

\authorrunning{~}

\def\subjclassHeading{} % pour suppprimer "2012 ACM Subject Classification"
%Theory of computation~Design and analysis of algorithms}


\let\oldsection\section
\let\oldsubsection\subsection
\def\mysection#1{\oldsection[\texorpdfstring{#1}{}]{\boldmath{#1}}}
\def\mysectionstar#1{\oldsection*{\boldmath{#1}}}
\def\mysubsection#1{\oldsubsection[\texorpdfstring{#1}{}]{\boldmath{#1}}}
\def\mysubsectionstar#1{\oldsubsection*{\boldmath{#1}}}
\makeatletter % pour \@ifstar
\def\section{\@ifstar{\mysectionstar}\mysection}
\def\subsection{\@ifstar{\mysubsectionstar}\mysubsection}
\makeatother

\begin{document}

\maketitle

\begin{abstract}
We consider the problem of finding the smallest graph that contains two input trees each with at most $n$ vertices preserving their distances. In other words, we look for an isometric-universal graph with the minimum number of vertices for two given trees. We prove that this problem can be solved in time $O(n^{5/2}\log{n})$. We extend this result to forests instead of trees, and propose an algorithm with running time $O(n^{7/2}\log{n})$. As a key ingredient, we show that a smallest isometric-universal graph of two trees essentially is a tree. Furthermore, we prove that these results cannot be extended. Firstly, we show that deciding whether there exists an isometric-universal graph with $t$ vertices for three forests is NP-complete. Secondly, we show that any smallest isometric-universal graph cannot be a tree for some families of three trees. This latter result has implications for greedy strategies solving the smallest isometric-universal graph problem.
\end{abstract}

%\newpage

\section{Introduction}
\label{sec:intro}

% Informal definition
Trees are one of the most natural and basic graph classes that can be studied. Studying problems on trees is thus natural, and is at the crux of the intersection between Graph Theory and Computer Science. Even for a general problem such as ``finding a graph that contains two input trees'', there are a multitude of variants. To cite few of them, there is the case where the final graph is constrained to be a tree, where ``containing'' can be either as a subgraph or as a topological embedding~\cite{GN98}. There is also the more general case where the input is two graphs instead of two trees~\cite{AbuKhzam14,AKBS17,BJK00}, or where the input is composed of more than two trees contained as minors~\cite{Bodini02,GKLPS18}.

In this paper, we consider the problem of constructing a graph containing as isometric subgraphs two given trees (or forests) with the aim of minimizing its number of vertices. An \emph{isometric subgraph} of a graph $G$ is a subgraph $H$ in which the distances between all the vertices of $H$ are preserved in $G$. This is equivalent to the notion of \emph{distance-preserving embedding} used by Winkler~\cite{Winkler83}, and then reworded as \emph{distance-preservers} \cite{Bodwin21,BCE06,CE06} with important applications to spanners and distance labeling schemes.

This basic problem can be extended in many ways. First of all, one can consider a family of graphs as input rather than two trees. This leads to the notion of \emph{universal} graphs, a notion that has been widely studied. More precisely, a graph $\sU$ is \emph{isometric-universal} for a graph family $\sF$ if every graph of $\sF$ appears as an isometric subgraph in $\sU$. In this line of research, $\sU$ and $\sF$ are often constrained. For instance, if $\sU$ and $\sF$ are limited to caterpillars\footnote{A tree whose vertices of degree at least two induce a path.}, Chung, Graham, and Shearer \cite{CGS81} proved that $n$-vertex caterpillars have an isometric-universal caterpillar with $O(n^2/\log{n})$ vertices, which is tight. And, for the family of all $n$-vertex graphs, Esperet, Gavoille, and Groenland~\cite{EGG21a} showed that there exists an isometric-universal graph with at most $3^{n + O(\log^2{n})}$ vertices, and with a lower bound of $2^{(n-1)/2}$.

The notion of isometric subgraph can also be extended to $k$-isometric subgraphs as follows. A subgraph $H$ of $G$ is called $k$-isometric if the distances between all the vertices of $H$ at distance at most $k$ in $G$ are preserved in $H$. This definition generalizes several classical notions of graph containment. Indeed, a subgraph is simply a $0$-isometric subgraph, an induced subgraph is a $1$-isometric subgraph, and a distance-preserving embedding or isometric subgraph is an $\infty$-isometric subgraph. This leads to the notion of \emph{$k$-isometric-universal graphs} for a given family of graphs. Note that sparse universal graphs ($k=0$) and small induced-universal graphs ($k=1$) are well-studied topics in Graph Theory, see for instance \cite{BILOx24,DEGJx21,EJM23} and the references therein. Thus, the initial problem we consider in this paper corresponds to the case where $k=\infty$ and $\sF$ is composed of two trees, and more generally two forests.%\\
%
% Cyril: jamais de \\ pour des questions de présentation de paragraphe. Eventuellement \medskip, sinon c'est qu'on veut mettre des \subsection{} ou autre \paragraph{}, ou alors changer \parskip au début

%------------------------------------------
\subsection{Our contribution} 

On the positive side, we prove that constructing a minimum isometric-universal graph for two forests, each with at most $n$ vertices, can be done in time $O(n^{7/2}\log{n})$. On the negative side, we show that this result cannot be extended to three forests by proving that the associated decision problem is NP-complete.

More precisely, in \cref{sec:2forests}, we demonstrate that a minimum isometric-universal graph for two trees can be constructed in time $O(n^{5/2}\log{n})$ thanks to an algorithm due to Gupta and Nishimura~\cite{GN98}. Additionally, we show how to reduce the problem for two forests from the \ASS Problem, a variant of the maximum matching of maximum weight problem in bipartite graphs, with an $O(n)$ slow-down factor in time complexity. As a key ingredient, we show that two trees always admit a minimum isometric-universal graph that is a tree. We believe that this nice property for two trees is an important technical contribution of the paper.

On the other hand, in \cref{sec:3forests}, we show that these positive results cannot be extended in at least two ways. First, we show that deciding whether there exists an isometric-universal graph with $t$ vertices for three forests is NP-complete. This is achieved by a technical reduction from the \MATCH Problem. Second, we exhibit an infinite family of triples of trees having no minimum isometric-universal graph that is a tree. Implications to approximation algorithms are discussed in \cref{subsec:greedy}. In particular, we demonstrate that no greedy strategy can build a minimum isometric-universal graph, even with access to an optimal oracle for a pair of graphs. Finally, we show that the property used for two trees cannot be extended to minimum $k$-isometric-universal graphs if $k < (n-8)/3$.

% \Cyril{Bon, j'aime pas le problème $k(n)$-iso-machin, car je ne sais pas bien ce qu'on peut arriver à dire. C'est comme le problème $k(n)$-coloring. $k(n)=3$ est NPC, et polynomial pour $k(n)=n$. Et alors ? Au mieux à mettre dans la conclus. Ce n'est pas une contribution en soit de définir ce problème là. Il n'a pas sa place ici.}

%Then, we establish the following parameterized decision problem: let $k:\mathbb{N} \rightarrow \mathbb{N}$:
%
%\myproblem%
%{$k(n)$-isometric-universal Graph of Two Graphs}%
%{($k(n)$-\IUGG)}%
%{Two graphs $G,H$ and $t \in\mathbb{N}$.}%
%{Let $n=\max\{|V(G)|, |V(H)|\}$. Is there a $k(n)$-isometric-universal graph for $\set{G,H}$\\ with $t$ vertices?}
%
%The previous results combined with this paper imply that $k(n)$-\IUGG is  NP-complete when $k(n)=1$, and that it is polynomial when $k(n)=n$. This raises the following question: for what function $k_{\mathcal{NP}}$ does this problem switch from NP-complete to polynomial ? We are quite convinced that $c \le k_{\mathcal{NP}}(n) \le n/2$ for every constant $c$, and we conjecture that $k(n)$-\IUGG is NP-complete if and only if $k(n)\le n^{1-1/c}$ for every positive integer $c$, and every sufficiently large $n$. \Francois{Est-ce qu'on préfère ça à la notation O ?}

%------------------------------------------
\subsection{Preliminaries and definitions}
% More formal definition

We consider finite and undirected graphs with no loops and no multi-edges. The \emph{union graph} of two graphs $X,Y$, denoted by $X\cup Y$, is defined by $V(X\cup Y) = V(X)\cup V(Y)$ and $E(X\cup Y) = E(X)\cup E(Y)$. Similarly, the \emph{intersection graph} of $X,Y$, denoted by $X\cap Y$, is defined by $V(X\cap Y) = V(X)\cap V(Y)$ and $E(X\cap Y) = E(X)\cap E(Y)$. The disjoint union is denoted by $X \cupdot Y$.

Given a graph $G$, and two vertices $u,v$ of $G$, we denote by $\dist_G(u,v)$ the minimum number of edges of a path connecting $u$ to $v$ in $G$, if it exists, and we set $\dist_G(u,v) = \infty$ otherwise, i.e., if $u$ and $v$ are in different connected components of $G$.

%\Cyril{Je pense qu'on pourrait dire que $|G|$ est un "short" pour $|V(G)|$ car cela allège un peu les preuves. Bon, on verra à la fin.}

A subgraph $H$ of $G$ is \emph{$k$-isometric} if, for all vertices $u,v$ of $H$, $\dist_H(u,v) = \dist_G(u,v)$ if $\dist_G(u,v) \le k$. Note that, in that case, if $\dist_H(u,v) > k$, then $\dist_G(u,v) > k$. In particular, since $H$ is a subgraph of $G$, if $\dist_H(u,v) = k+1$, then $\dist_G(u,v) = k+1$.

%\Cyril{Faut faire attention au point suivant. On a tendance à dire que le graphe $H$ est $k$-isometric ou pas dans $G$. C'est juste que c'est pas vrai. Considérons $H$ un chemin à trois sommets, et $G$ une maison (carré + triangle partageant une arête). On prend $k=1$. Il se trouve que $H$ apparaît comme sous-graphe $k$-isométric (induit donc) de $G$ (dans le carré), mais qu'il y a un plongement de $H$ dans $G$ qui ne l'est pas (dans le triangle). Bref, quand on fixe l'embedding de $H$ dans $G$, disons $\hH$, alors on oui: $\hH$ est $k$-isometric ou pas dans $G$. Le problème intervient surtout quand on veut dire que $H$ n'est pas un sous-graphe $k$-isométric de $G$. Dire que dans le triangle ne marche pas ne suffit pas: il faut le démontrer pour tout plongement, bien évidemment.}

% \Cyril{En soit $k$ n'est pas forcément un entier, en tout cas il faut faire attention que d'écrire $k \in \mathbb{N}$ exclut le cas $k = \infty$ et le cas $k = n/2$ lorsque $n$ est impair.} \Edgar{On peut définir au début $\mathbb{K} = \mathbb{N}\cup\set{\infty}$}.

A \emph{$k$-isometric-universal graph} for a given graph family $\sF$ is a graph $\sU$ such that every graph $G$ of $\sF$ is isomorphic to a $k$-isometric subgraph of $\sU$. Moreover, $\sU$ is called \emph{minimum} if it has the smallest possible number of vertices. $\sU$ is \emph{minimal} if no proper subgraph of $\sU$ is a $k$-isometric-universal graph for $\sF$. Note that if $\sU$ is minimal, then, for every $k$-isometric embedding $\varphi$ from the graphs of $\sF$ to $\sU$, we have $\sU = \bigcup_{G \in \sF} \varphi(G)$.
% Cyril: sinon, on pourrait enlever la différence symétrique de $\sU$ et donc trouver un sous-graphe universel strict à $\sU$
Whenever $k = \infty$, we rather use the term \emph{isometric-universal graph} as introduced by Esperet et al. \cite{EGG21a}.

%%%%%%%%%%%%%%%%%%%%%%%%%%%%%%%%%%%%%%%%%%%%%%%%%%%%%%%%%%%%%%%%%%%%%% 

\section{Universal Graph for Two Forests}
\label{sec:2forests}

%\Cyril{J'ai sucré "Isometric" dans les titres de section. L'idée est que dans tout le papier on parle d'isométrique (et aussi de $k$-isométric). Je pense qu'il n'y a pas d'ampiguité possible. J'ai aussi mis "for trees" à la place de "of trees", car c'est ce qu'on fait dans le titre du papier, et j'aime mieux.}

Let us first remark that computing the minimum isometric-universal graph is hard in general, even for two graphs. More precisely, given two graphs $G,H$, and some number $t$, the $k$-\IUGG Problem asks for finding a $k$-isometric-universal graph with $t$ vertices for $\set*{G,H}$, that is a $t$-vertex graph $\sU$ containing $G$ and $H$ as $k$-isometric subgraphs.

% we consider the following decision problem:
%
%\myproblem%
%{$k$-isometric-universal Graph of Two Graphs}%
%{($k$-\IUGG)}%
%{Two graphs $G,H$ and $t \in\mathbb{N}$.}%
%{Is there a $k$-isometric-universal graph for $\set{G,H}$ with $t$ vertices?}
%\newpage

Indeed, the problem is NP-complete for each $k\in\mathbb{N}$ and for two graphs with the same number of vertices. For $k=0$, this is simple reduction\footnote{It is easy to check that the problem is in NP, the certificate being composed of the universal graph and of the embedding from each graph to the universal graph, both can be described polynomially in the size of $G,H$.} from \HP, and for $k>0$, this is a reduction from \KP of size $p$. Indeed, for a graph $G$ with $n$ vertices, set $H$ to a path (if $k=0$) or to a clique (for $k>0$) with $n$ vertices, and ask whether the there exists a $k$-isometric-universal graph for the family $\set{G,H}$, with respectively $t = n$ vertices for $k=0$, or with $t = 2n-p$ vertices for $k>0$. 

We will consider the variant where the pair $\set{G,H}$ is composed of two forests. The problem is known to be NP-complete for $k=0$~, even if one of the two forests is a tree~\cite[Theorem~4.6, p.~105]{GJ79}. As shown by Gavoille and Jacques~\cite{GJ24}, this NP-completeness result can be extended to two trees of the same size for $k=1$. Nevertheless, as we will see in this section, the problem becomes polynomial for two forests if $k=\infty$.

%\Cyril{On a une jolie définition de problème $k$-\IUGG dont on ne resert plus après ... En fait, seulement pour deux forêts et $k=\infty$. Bon, on verra à la fin.}

%\Cyril{Question naturelle à mettre en conclusion. Quel est le $k = k(n)$ tel que cela bascule entre NP-Complet et polynomial ? Qu'en est-il pour trois arbres ou trois forêts ?}

%\Cyril{Dire peut-être que c'est polynomial pour $t$ forêts linéaires. Pour $t=2$, c'est comme deux forêts, et pour plus de forêts linéaires la remarques est qu'on ne peut pas faire mieux que de fusionner les composantes max. Dire qu'en fait c'est linéaire. À voir.}

%------------------------------------------
\subsection{Connected components reduction}

%\Status{Cyril: Passe faite sur tout le \cref{th:2graphs}.}

We start with the following result, that will allow us to reduce the problem of computing the minimum isometric-universal graph for a pair of graphs (possibly non-connected) to the case of their connected components. The reduction is based on a polynomial solution of the \ASS Problem whose goal is to optimally assign agents to individual tasks. For the purpose of the next statement, we say that a function $f$ is \emph{superlinear} if $f(n)/n$ is non-decreasing. Typically, polynomials $n^{1+c}$ and exponentials $2^{cn}$ for $c>0$ are superlinear functions.

\begin{theorem}\label{th:2graphs}
    Let $G,H$ be two graphs with at most $n$ vertices, composed respectively of $s$ and $r$ connected components. Let $f(m)$ be any superlinear function bounding above the time complexity of computing a minimum isometric-universal graph for any pair of graphs with at most $m$ vertices composed of a component of $G$ and of $H$. Then, assuming $s\le r$, 
    % Cyril: "s" pour small
    a minimum isometric-universal graph for $\set{G,H}$ can be computed in time %\Edgar{C'est quoi $n$ ?} \Cyril{C'est le nombre max de sommets pour $G$ ou $H$.}
    \[
        O\pare{ r\min\set{f(2n),\, s f(n)} + s \min\set{ r \sqrt{s}\log{n},\, rs + s\log\log{s}} } ~.
    \]
\end{theorem}

\begin{proof}
    Let $G_1,\dots,G_s$ and $H_1,\dots,H_r$ be the connected components of $G$ and $H$ respectively, where $s\le r$. Let $ n = \max\set{|V(G)|,|V(H)|}$. For convenience, let $I(s,r) = \range{1}{s}\times\range{1}{r}$, and let $X = \range{G_1}{G_s}$ and $Y = \range{H_1}{H_r}$ be the sets of components of $G$ and $H$ respectively. 

    For each $(i,j) \in I(s,r)$, let $S(G_i,H_j)$ be any minimum isometric-universal graph for the (ordered) pair $(G_i,H_j)$.
    % Cyril: En fait, la phrase suivante trouble un peu, car dans les faits on ne va jamais utiliser $S(H_j,G_i)$.
    % The graphs $S(G_i,H_j)$ and $S(H_j,G_i)$ are not necessarily isomorphic, however they must have the same number of vertices.
    We also define the weight function $w(G_i,H_j) = \cardV*{G_i} + \cardV*{H_j} - \cardV*{S(G_i,H_j)}$, corresponding to the number of vertices that $G_i$ and $H_j$ have in common in $S(G_i,H_j)$. We have $w(G_i,H_j) \in \range{1}{n}$. %\Edgar{C'est quoi n ?}.
    The graph $S(G_i,H_j)$, and thus its weight $w(G_i,H_j)$, can be computed in time at most $f(\max\set{\cardV*{G_i}, \cardV*{H_j}})$  by definition of $f$.

   % \Cyril{Il est important de prendre l'union disjointe, car si $G$ et $H$ ont des composantes communes (par exemple imaginons que $K_{1,3} \in G$ et $K_{1,3} \in H$), alors sans l'union disjointe, formellement $B$ ne sera plus isomorphe à un biparti comple et cela va rendre plus délicate la suite.} \Edgar{Cela ne change rien vu que les sommets sont des noms des graphes et non pas des graphes non ?}
    
    Now, we construct an auxiliary edge-weighted bipartite graph $B = (X \cupdot Y, X \times Y, w)$, the weight of each edge $G_i-H_j$ being set to $w(G_i,H_j)$, for $(G_i,H_j) \in X\times Y$. Note that $B$ is isomorphic to the complete bipartite graph $K_{s,r}$. See \cref{fig:graphB}.

    \begin{myfigure}
    \begin{tikzpicture}

    % les sommets
    %\tikzset{every node/.style={vertex}}

    \def\h{2}
    \node (G1) at (0.5,\h) {$G_1$};
    \node (x1) at (1.5,\h) {$\cdots$};
    \node (Gi) at (2.5,\h) {$G_i$};
    \node (x2) at (3.5,\h) {$\cdots$};
    \node (Gs) at (4.5,\h) {$G_s$};
    
    \def\h{0}
    \node (H1) at (0,\h) {$H_1$};
    \node (H2) at (1,\h) {$H_2$};
    \node (y1) at (2,\h) {$\cdots$};
    \node (Hj) at (3,\h) {$H_j$};
    \node (y2) at (4,\h) {$\cdots$};
    \node (Hr) at (5,\h) {$H_r$};
    
    % les arêtes
    \graph{ (G1) -- { (H1), (H2), (Hj), (Hr) } };
    \graph{ (Gi) -- { (H1), (H2), (Hj), (Hr) } };
    \graph{ (Gs) -- { (H1), (H2), (Hj), (Hr) } };

    \end{tikzpicture}
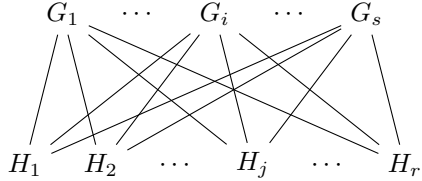
    \caption{The edge-weighted complete bipartite graph $B$ defined on the $s$ components of $G$ and the $r$ components of $H$, where the weights are $w(G_i,H_j) = \cardV*{G_i} + \cardV*{H_j} - \cardV*{S(G_i,H_j)}$.}
    \label{fig:graphB}
    \end{myfigure}

    Let $\sK = \set{ S(G_i,H_j) : (i,j)\in I(s,r) }$ be the family of all the graphs $S(G_i,H_j)$ with $(G_i,H_j) \in X\times Y$.
    
    \begin{claim}\label{claim:compute_K_sr}
        The graph family $\sK$, and the graph $B$ with its edge-weights, can be constructed in time $O(r \min\set{f(2n), s f(n)})$. 
     \end{claim}

    %\Cyril{On peut pas faire bien mieux que $r f(n)$ à cause de l'exemple suivant: $s=1$ avec un graphe $G_1$ de $n$ sommets, et $H_1,...,H_r$ chacun avec au plus $n$ sommets. On a $1 \le s \le r$. Calculer le graphe $S(G_1,H_j)$ prend clairement un temps $f(\max\set{|V(G_1)|,|V(H_j)|}) = f(n)$. Et malheureusement, on doit répéter ce calcul $r$ fois, soit un total de $r f(n)$ pour les $rs$ graphes. Pas facile non plus d'enlever le $2n$ dans $f(2n)$, même en disant que $f(m)$ est la complexité pour deux componsantes avec $m$ sommets au total (voir preuve).}
    
    \begin{proof}
        To construct all the $|X|\cdot|Y| = rs$ graphs of $\sK$, it suffices to consider each edge $G_i-H_j$ of $B$ and to construct the graph $S(G_i,H_j)$. This takes a time at most: 
        \begin{equation*}
            \sum_{(i,j)\in I(s,r)} f\pare{\max\set{\cardV{G_i},\cardV{H_j}}} ~\le~ |I(s,r)| \cdot f(n) = rs\cdot f(n)
        \end{equation*}
        using the fact that $f$ is non-decreasing. Indeed, for $x\le y$, we have, by superlinearity of $f$, that $f(x)/x \le f(y)/y \le f(y)/x$, i.e., $f(x) \le f(y)$. To evaluate this time consumption more accurately, we order the vertices of $B$ as follows. Let $Z_1, Z_2, \dots, Z_{r+s}$ be the sequence of all the connected components of $X\cupdot Y$ (so the vertices of $B$), ordered by their number of vertices in a non-increasing way. So $Z_1$ is the largest component of $X \cupdot Y$, $Z_2$ the largest one taken in $(X\cupdot Y)\setminus Z_1$, $Z_3$ the largest in $(X\cupdot Y)\setminus (Z_1\cup Z_2)$, and so on. Obviously, $\sum_{i=1}^{r+s} \cardV*{Z_i} = \cardV*{G} + \cardV*{H} \le 2n$.
        % \Cyril{On va se taper le $2n$ à cause de cette équation. Notons que $f(2n) = O(f(n))$ est évidemment faut en général. Par exemple, si $f(n) = 2^n$. On a $f(2n) = f(n)^2$. Donc on ne peut pas simplifier.}

        From this order, all the graphs of $\sK$ can be computed in $r+s$ steps as follows. First, consider $Z_1$, and either compute $S(Z_1,H_t)$ for all $H_t \in Y$ if $Z_1 \in X$, or compute $S(G_t,Z_1)$ for all $G_t \in X$ if $Z_1 \in Y$. Depending on whether $Z_1$ is in $X$ or in $Y$, this takes a time $\sum_{t=1}^r f(\cardV*{Z_1})$ or $\sum_{t=1}^s f(\cardV*{Z_1})$, since $Z_1$ is the largest graphs of $X\cup Y$. This is at most $r \cdot f(\cardV*{Z_1})$ by definition of $r$, and non-decreasing property of $f$. Then, we consider $Z_2$. Note that we do not need to compute either $S(Z_1,Z_2)$ or $S(Z_2,Z_1)$. This is because either $Z_1,Z_2$ are in the same part (in $X$ or in $Y$) (and so the edge $Z_1 - Z_2$ does not exist in $B$), or the edge $Z_1 - Z_2$ does exist in $B$, and thus $S(Z_1,Z_2)$ or $S(Z_2,Z_1)$ has already been computed by the first step. So, one computes either $S(Z_2,H_t)$ for all $H_t \in Y\setminus Z_1$ if $Z_2 \in X$, or $S(G_t,Z_2)$ for all $G_t \in X\setminus Z_1$ if $Z_2 \in Y$. This takes a time at most $r \cdot f(\cardV*{Z_2})$ since $Z_2$ is larger (or equal) than every component in $(X\cupdot Y) \setminus Z_1$.
        
        In general, at the $i$-th step, $Z_i$ is considered, and one computes:
        \begin{itemize}%[noitemsep]
        
            \item $S(Z_i,H_t)$ for all $H_t \in Y\setminus (Z_1 \cup \cdots \cup Z_{i-1})$ if $Z_i \in X$, or 
        
            \item $S(G_t,Z_i)$ for all $G_t \in X\setminus (Z_1 \cup \cdots \cup Z_{i-1})$ if $Z_i \in Y$.
        
        \end{itemize}
        Step $i$ takes a time at most $r \cdot f(\cardV*{Z_i})$, since $Z_i$ is larger (or equal) than every graph in $(X\cupdot Y) \setminus (Z_1 \cup \cdots \cup Z_{i-1})$, and since $\max\set{|X|,|Y|} \le r$.
        
        By this way, all edges of $B$ has been considered, and the construction time for $\sK$ is at most:
        \begin{align*}            
            \sum_{i=1}^{r+s} r\cdot f\pare{\cardV{Z_i}} ~=~ r \sum_{i=1}^{r+s} f\pare{\cardV{Z_i}} ~.
        \end{align*}
        To conclude, we will use the fact that, for all positive numbers $x_1,\dots,x_t$,
        \[
            \sum_{i=1}^t f(x_i) ~\le~ f\pare{\sum_{i=1}^t x_i} ~.
        \]
        Indeed, let $m = \sum_{i=1}^t x_i$. Then, by superlinearity, $f(x_i)/x_i \le f(m)/m$, because $x_i \le m$ since $x_i > 0$. Therefore, $\sum_{i=1}^t f(x_i) \le \sum_{i=1}^t x_i \cdot f(m)/m = f(m) = f(\sum_{i=1}^t x_i)$.

        In total, the time to construct all the graphs of $\sK$ is at most:
        \begin{align*}            
            r \sum_{i=1}^{r+s} f\pare{\cardV{Z_i}} 
            ~\le~ r \cdot f\pare{ \sum_{i=1}^{r+s} \cardV{Z_i}}
            ~\le~ r \cdot f(2n) ~.
        \end{align*}
        Hence, the time to construct $\sK$, and $B$ with its edge-weights, is at most\\ $O(\min\set{r\cdot f(2n), rs\cdot f(n) }) = O(r\min\set{ f(2n), s f(n)})$ as claimed.
    \end{proof}

    Given $M \subseteq E(B)$, the \emph{weight} of $M$, denoted by $w(M)$, is the sum of its weights, i.e., $w(M) = \sum_{(G_i,H_j) \in E(M)} w(G_i,H_j)$.
    
    \begin{claim}\label{claim:matching}
        Given $\sK$ and any matching $M$ of $B$, one can construct in time $O(n)$ an isometric-universal graph $\sU_M$ for $\set{G,H}$ with $\cardV*{G} + \cardV*{H} - w(M)$ vertices.
    \end{claim}

    \begin{proof}
        The graph $\sU_M$ is composed of $G\cupdot H$ in which all the (ordered) pairs of components $(G_i,H_j) \in M$ are replaced by $S(G_i,H_j)$. Given $\sK$, which can be organized as an array indexed by $(i,j)$, $\sU_M$ can be constructed in time $O(|M| + \cardV*{\sU_M})$ by extracting each graph $S(G_i,H_j)$ from $\sK$ such that $(G_i,H_j)\in M$. As $|M| \le s \le n$, and $\cardV*{\sU_M} \le \cardV*{G} + \cardV*{H} \le 2n$, this is at most $O(n)$.
        
        It is easy to check that $G$ and $H$ are isometric subgraphs of $\sU_M$, since each $G_i$ and $H_j$ are isometric subgraphs of $\sU_M$, and since the distance between any two vertices of different components of $G$ (or of $H$) in $\sU_M$ is $\infty$ as required.

        From the previous construction, the number of vertices of $\sU_M$ is:
        \begin{align*}
            \cardV{\sU_M} &= \cardV{G} + \cardV{H} - \pare{ \sum_{(G_i,H_j)\in M} \cardV{G_i} + \cardV{H_j} } \\
            &+ \pare{ \sum_{(G_i,H_j)\in M} \cardV{S(G_i,H_j)}}\\
            &= \cardV{G} + \cardV{H} - \pare{ \sum_{(G_i,H_j)\in M} \cardV{G_i} + \cardV{H_j} - \cardV{S(G_i,H_j)} }\\
            &= \cardV{G} + \cardV{H} - \sum_{(G_i,H_j)\in M} w(G_i,H_j)\\
            &= \cardV{G} + \cardV{H} - w(M) ~.
        \end{align*}
    \end{proof}

    As suggested by \cref{claim:matching}, minimizing $|V(\sU_M)|$ can be achieved by maximizing $w(M)$. So, it remains to show that a minimum isometric-universal graph for $\set{G,H}$ can be constructed from a maximum matching of maximum weight for $B$, that is a matching of maximum weight among the matchings of maximum size.

%\Francois{Attention : Perfect matching = tous les sommets sont couverts. Ici c'est un maximum matching non pondéré, qu'on oppose à un maximum matching pondéré.}
    
    For this purpose, consider any isometric-universal graph $\sU$ for $\set{G,H}$, and let $U_1, U_2,\dots$ be its components. We consider any isometric embedding $\varphi$ of $G$ and $H$ in $\sU$, and  we denote by $\hG,\hH$ the corresponding isometric subgraphs of $G$ and $H$ in $\sU$. Similarly, $\hG_i,\hH_j$ denote the isometric subgraphs of the components $G_i,H_j$ embedded in $\sU$ w.r.t. $\varphi$.

    We observe that:
    
    \begin{claim}\label{claim:two_cc}
        Each component of $\sU$ contains at most one component of~$G$ and one component of~$H$.
    \end{claim}
    
    \begin{proof}
        Assume that some component $U_k$ contains two components $\hG_i,\hG_j$ of $G$ (similarly for two components of $H$). Consider $u,v$ be two vertices of $G$ taken respectively in $G_i,G_j$. We have $\dist_G(u,v) = \infty$. But, since $\hG_i,\hG_j \subseteq U_k$, $\dist_{\sU}(u,v) = \dist_{U_k}(u,v) < \infty$. Therefore, $\hG$ is not an isometric subgraph of $\sU$: a contradiction with the choice of $\varphi$.
    \end{proof}

    Now, assume that $\sU$ is minimal, so that $\sU = \hG\cup \hH$. Combined with \cref{claim:two_cc}, for each component $U_k$ of $\sU$, only one of the three following cases occurs:

    %\Edgar{TODO: changer les cas pour que ça soit une disjonction (cas 3 proritaire)}
    \begin{itemize}%[noitemsep]
        \item Case $C_1$: $U_k = \hG_i \cup \hH_j$,
        \item Case $C_2$: $U_k = \hG_i$ and Case $C_1$ does not occur, or
        \item Case $C_3$: $U_k = \hH_j$ and Case $C_1$ does not occur.
    \end{itemize}
    Since one cannot have $U_k = \hG_i$ and $U_k = \hH_j$ without Case $C_1$, the three cases are indeed disjoint.
    
    Let $D = \set*{ (G_i,H_j) \in X\times Y: \exists k, U_k = \hG_i\cup\hH_j }$ be the set of all pairs of components for which Case $C_1$ holds. We have $D \subseteq E(B)$.
    
    \begin{claim}\label{claim:minimum}
        If $\sU$ is minimum and minimal, then $D$ is a maximum matching of $B$, and $\cardV*{\sU} = \cardV{G} + \cardV{H} - w(D)$.
    \end{claim}

    \begin{proof}
        If $D$ is not a matching, then we have for some $i'\neq i$ or $j'\neq j$:
        $$
          (G_i,H_j),(G_i,H_{j'}) \in D \mbox{ or } (G_i,H_j),(G_{i'},H_j) \in D.
        $$
        It means that $\hG_i,\hH_j,\hH_{j'}$ or $\hG_i,\hG_{i'},\hH_{j}$ are contained in a same component $U_k$ of $\sU$: a contradiction with \cref{claim:two_cc}.
        
        If the matching $D$ is not maximum, then $|D| < s$, since $B$ is (isomorphic to) a complete bipartite graph $K_{s,r}$ with smallest part of size $s$. So we can select from $E(B)$ a pair $(G_i,H_j) \notin D$, i.e., such that Case $C_1$ does not hold. In particular, there must exist two indices $k_i\neq k_j$ for which $U_{k_i} = \hG_i$ (Case $C_2$) and $U_{k_j} = \hH_j$ (Case $C_3$). We can construct a new minimal isometric-universal graph for $\set{G,H}$ obtained from $\sU$ by merging $U_{k_i}$ and $U_{k_j}$ by one vertex, so saving one vertex from $\sU$: a contradiction with the fact that $\sU$ was already minimum.

        Let us calculate the number of vertices of $\sU$. From the classification of $\sU$'s components, three cases occur. For each $i \in\set{1,2,3}$, we denote by $K_i$ the set of $U_k$'s component indices $k$ such that Case $C_i$ occurs.
        
        By case disjunction, we have:
        \[
            \cardV{\sU} = \sum_{k} \cardV{U_k} = \sum_{k \in K_1} \cardV{U_k} + \sum_{k\in K_2} \cardV{U_k} + \sum_{k \in K_3} \cardV{U_k} ~.
        \]
        
        If $k \in K_2$, then $\cardV*{U_k} = \cardV*{\hG_i} = \cardV*{G_i)}$. Similarly, if $k \in K_3$, then $\cardV*{U_k} = \cardV*{\hH_j} = \cardV*{H_j}$. 
        By definition of $D$, if $(G_i,H_j) \in D$, then Case $C_1$ occurs for $\hG_i$ and $\hH_j$, and by case disjunction, there is no index $k$ for which Case $C_2$ occurs for $\hG_i$ or Case $C_3$ occurs for $\hH_j$. Thus, the sum for $k\in K_2\cup K_3$ can be rewritten as:
        
        \begin{align*}
            \sum_{k \in K_2 \cup K_3} \cardV{U_k} &= \sum_{i=1}^s \cardV{G_i} + \sum_{j=1}^r \cardV{H_j} - \sum_{(G_i,H_j) \in D} (\cardV{G_i} + \cardV{H_j}) \\
            &= \cardV{G} + \cardV{H} - \pare{ \sum_{(G_i,H_j) \in D} \cardV{G_i} + \cardV{H_j} } ~.
        \end{align*}
        
        Consider some $k \in K_1$, i.e., $U_k = \hG_i \cup \hH_j$ for some $(i,j)$, which is equivalent to say that $(G_i,H_j) \in D$. The minimality of $\sU$ (in terms of number of vertices) implies that $U_k$ has to be minimum for $\set{G_i,H_j}$, i.e., $\cardV*{U_k} = \cardV*{\hG_i \cup \hH_j} = \cardV*{S(G_i,H_j)}$. Otherwise, one could decrease $\cardV*{\sU}$ by simply replacing $U_k$ by $S(G_i,H_j)$. Thus, the sum for $k\in K_1$ can be rewritten as:
        \[
            \sum_{k \in K_1} \cardV{U_k} = \sum_{(G_i,H_j) \in D} \cardV{S(G_i,H_j)} ~.
        \]

        It follows that:
        \small{
        \begin{align*}
            \cardV{\sU} &= \sum_{ k \in K_2 \cup K_3} \cardV{U_k} + \sum_{(G_i,H_j) \in D} \cardV{S(G_i,H_j)} \\
            &= \cardV{G} + \cardV{H} - \pare{ \sum_{(G_i,H_j) \in D} \cardV{G_i} + \cardV{H_j}} + \sum_{(G_i,H_j) \in D} \cardV{S(G_i,H_j)} \\
            &= \cardV{G} + \cardV{H} - \pare{ \sum_{(G_i,H_j) \in D} \cardV{G_i} + \cardV{H_j} - \cardV{S(G_i,H_j)} } \\
            &= \cardV{G} + \cardV{H} - \sum_{(G_i,H_j) \in D} w(G_i,H_j) \\
            & = \cardV{G} + \cardV{H} - w(D)
        \end{align*}}
        by definition of the weight function $w$.
    \end{proof}

    It is easy to check that the maximum matching $D$, if $\sU$ is minimum and minimal as in \cref{claim:minimum}, is also of maximum weight. Indeed, by way of contradiction, assume $B$ has a maximum matching $M$ with $w(M) > w(D)$. By \cref{claim:minimum}, $\cardV*{\sU} = \cardV*{G} + \cardV*{H} - w(D)$. Consider the universal graph $\sU_M$ as in \cref{claim:matching}. It has $\cardV*{\sU_M} = \cardV*{G} + \cardV*{H} - w(M)$ vertices. So, if $w(M) > w(D)$, then $\cardV*{\sU_M} < \cardV*{\sU}$: a contradiction with the fact that $\sU$ is minimum.

    To summarize, \cref{claim:matching} and \cref{claim:minimum} tell us that the minimum number of vertices for an isometric-universal graph for $\set{G,H}$ is $\cardV*{G} + \cardV*{H} - w(M)$, where $M$ is a maximum matching for $B$ with maximum weight. Moreover, given $\sK$ and $M$, such a graph can be constructed in time $O(n)$ by \cref{claim:matching}. It remains to compute $M$, given the weighted bipartite graph $B$.

    \def\sB{\mathscr{B}} 

    % Cyril
    %
    % La question qui se pose est la mise en {\oe}uvre de cet algorithme pour le rendre efficace. L'algorithme d'Edmonds~\cite{Edmonds65}, est basé sur le même principe mais inclut la contraction de \textit{blossom} (cycle alternant de longueur impaire). D'autres algorithmes, dont ceux d'Even et Kariv en \annee{1975}, de Bartnik en \annee{1978}, ont une complexité un peu meilleure en $O\pare{ \min\set{n^{5/2},m\sqrt{n}\log{n}} }$. Tous ces algorithmes sont réputés pour être très complexes.

    The \ASS Problem asks for a maximum matching of maximum weight in a bipartite graph $\sB$ with integral edge weights. According to~\cite{RT12a}, this problem can be solved in time $O(m \sqrt{t} \log{(t C)})$, where $m$ is the number of edges of $\sB$, $t$ the size of the maximum matching, and $C$ an upper bound on the edge weights. Using the Thorup's technique~\cite{Thorup04b}, in conjunction with~\cite{RT12a}, the computation can also be done in time $O(mt + t^2\log\log{\rho})$, where $\rho$ is the size of the smallest part of $\sB$.

    For our purpose, $\sB = B$, and thus $m = rs$, $t = \rho = s$, and $C = n$, since $B$ is isomorphic to $K_{s,r}$, and since edge-weights are taken in $\range{1}{n}$. Therefore, the computation of a maximum matching $M$ of maximum weight for $B$, given $B$, can be done in time:
    \begin{align*}
        & O\pare{\min\set{ m\sqrt{t}\log{(tC)},\, mt+t^2\log\log{\rho}}} \\
        &= O\pare{ \min\set{ r s^{3/2}\log{n}, \, rs^2 + s^2 \log\log{s}} } \\
        &= O\pare{ s \min\set{ r \sqrt{s}\log{n}, \, rs + s \log\log{s}} }
    \end{align*}

    From \cref{claim:compute_K_sr}, $\sK$, and thus $B$, can be computed in time $O(r\min\set{f(2n), s f(n)})$. Overall, the time to construct $B$, $M$, and $\sU_M$, which is a minimum isometric-universal graph for $\set{G,H}$, is:
    \[
        O\pare{ n + r\min\set{f(2n),\, s f(n)} + s \min\set{ r \sqrt{s}\log{n},\, rs + s \log\log{s}} } ~.
    \]
    The leftmost term ''$n+$'', coming from \cref{claim:matching} and the construction of $\sU_M$, can be omitted from the overall complexity by remarking that, by superlinearity, $f(n)/n \ge f(1)/1$, which implies $f(n) = \Omega(n)$. Indeed, as $f(n)$ is an upper bound on a time complexity, $f(1) > 0$ holds necessarily. Hence, the term $r \min\set{f(2n),s f(n)} = \Omega(n)$.
    
    This completes the proof of \cref{th:2graphs}.
\end{proof}

%------------------------------------------
\subsection{The case of two trees}

To compute the minimum isometric-universal graph of two forests in polynomial time, we will combine \cref{th:2graphs} and the next result about two trees.

%\Status{Edgar \& Cyril \& François: Passe faite sur ce théorème}

%\Cyril{Voir le problème avec "minimum" qui n'implique pas "minimal" (ils ne concernent pas la même relation: nombre de sommets vs. inclusion de sous-graphes) ... Solution: "Smallest = minimum + minimal" ??? On verra à la fin.}

\begin{theorem}\label{th:2trees}
  Every minimum and minimal isometric-universal graph for two trees is a tree.
\end{theorem}

Before presenting its technical proof, we remark that both properties, \emph{minimum} and \emph{minimal}, are required. % for \cref{th:2trees}.
Indeed, the disjoint union of any pair of trees is isometric-universal for this pair, is minimal but not minimum (as implied by \cref{th:2trees}, the minimum one must be connected). Moreover, the graph composed of two pendant vertices attached to some vertex of a $C_4$ (see \cref{fig:minimal-minimum}) is a minimum isometric-universal graph for the family $\set*{K_{1,4}, K^+_{1,3}}$ composed of two $5$-vertex trees, where $K^+_{1,p}$ denotes here the graph $K_{1,p}$ where one edge has been subdivided into two edges. However, it is not minimal: it contains $K^+_{1,4}$ as proper subgraph which is isometric-universal for $\set*{K_{1,4}, K^+_{1,3}}$. This example can be generalized to two trees with $n$ vertices, for every $n\ge 5$.

\begin{myfigure}
\begin{tikzpicture}

    \node[text width=3cm] at (1.75,0.65) {$e$};

    % les sommets
    \tikzset{every node/.style={vertex}}
    \node (1) at (0, 0) {};
    \node (2) at (1,+1) {};
    \node (3) at (1,-1) {};
    \node (4) at (2, 0) {};
    \node (5) at (3,+1) {};
    \node (6) at (3,-1) {};
    
    % les arêtes
    \graph{ (4) -- { (2), (3), (5), (6) } };
    \graph{ (1) -- { (2), (3) } };

\end{tikzpicture}
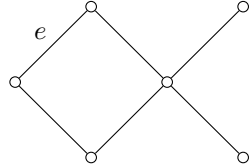
\caption{A minimum isometric-universal graph for the $5$-vertex tree family $\set*{K_{1,4}, K^+_{1,3}}$. From \cref{th:2trees} it is not minimal (the edge $e$ can be removed).}
\label{fig:minimal-minimum}
\end{myfigure}

\begin{proof}[Proof of \cref{th:2trees}]
  Let $\sU$ be any minimum and minimal isometric-universal graph for $\sF = \set{T_1, T_2}$ composed of two trees. We consider a given isometric embedding of $T_1$ and $T_2$ in $\sU$, and denote by $A_1$ and $A_2$ their corresponding embedded subgraphs in $\sU$. Since $\sU$ is minimal, $\sU = A_1 \cup A_2$. So, each vertex or edge of $\sU$, belongs either to $A_1\setminus A_2$, $A_2\setminus A_1$, or to $A_1\cap A_2$.

  For convenience, let us denote by $d(x,y) = \dist_{\sU}(x,y)$, and by $|P|$ the \emph{length} of path $P$, i.e., its number of edges. We will often use the property that every path of $P$ of $A_i$, for some $i\in\set{1,2}$, satisfies $|P| = \dist_{A_i}(u,v) = d(u,v)$, where $u,v$ are the extremities of $P$. This is because every path of a tree is a shortest path in this tree, and because $A_i$ is an isometric subgraph of $\sU$.

  \definecolor{myred}{rgb}{0.9,0,0}
  \definecolor{mygreen}{rgb}{0,0.6,0}
  \definecolor{mygray}{rgb}{0.8,0.8,0.8}

  %By \cref{lem:connexe}, $\sU$ must be connected.
  $\sU$ must be connected, otherwise we could merge two connected components of $\sU$ by a vertex, and this without modifying any distance in $A_1$ or $A_2$ which are finite. This would contradicting the fact that $\sU$ is minimum. If $\sU$ is a tree, we are done. It remains to show that if $\sU$ has a cycle, then $\sU$ is not minimum. Consider the shortest cycle $C$ in $\sU$. Let us show that:

  \begin{claim}\label{claim:uv_P1P2}
      The edges of $C$ can be partitioned into two paths $P_1$ and $P_2$ such that (see \cref{fig:cyclereduccontrac} on the left):
      \begin{itemize}%[noitemsep]
          \item $P_1 \cap P_2 = \set{u,v}$ with $u,v \in A_1\cap A_2$
          \item $P_1\setminus P_2 \subset A_1\setminus A_2$
          \item $P_2\setminus P_1 \subset A_2\setminus A_1$
      \end{itemize}
  \end{claim}
   
  \begin{proof}
      Consider the connected components of %\footnote{The intersection graph $A_1\cap A_2$ is defined by $(V(A_1)\cap V(A_2),E(A_1)\cap E(A_2))$.}
      $A_1 \cap A_2$. Obviously, such components are trees.
      % Cyril: C'est parce que n'importe quoi intersecté avec un arbre donnera un sous-graphe qui est un arbres, et donc des arbres comme cc.
      In particular, $C \subset A_1\cap A_2$ is impossible. It is also clear that $C$ intersects at least two such components. Indeed, first, if it intersects no components, then $C$ is wholly contained in $A_1$ or in $A_2$, contradicting they are trees. Secondly, if it intersects only one, then $C$ is wholly contained in $A_1$ or in $A_2$ as well, since the part of $C$ outside the component exists and belongs to one tree, and the part inside the component belongs to both trees.

      Define $u$ and $v$ as any two vertices of $C$ that belong to two different components of $A_1\cap A_2$, and such that $\dist_C(u,v)$ is minimum. Let $P_1$ be a path of $C$ realizing this distance. By construction, the only vertices of $P_1$ that are in $A_1 \cap A_2$ are precisely its extremities $u,v$. Indeed, if $P_1\setminus\set{u,v}$ contains a vertex $w$ in $A_1\cap A_2$, then both $\dist_C(u,w)$ and $\dist_C(v,w)$ would be less than $\dist_C(u,v)$, which is not possible by the choice of $u,v \in C$. Thus, $P_1\setminus\set{u,v}$ contains only vertices and edges of a same tree $A_i$, because if two consecutive edges belong to different trees, their common vertex must be in the intersection. W.l.o.g., assume that $P_1 \subset A_1$. Note that we have $\dist_C(u,v) = d_{A_1}(u,v) = d(u,v) = |P_1|$ because $P_1$ is a path of $A_1$. We set $P_2$ to be the path of $C$ induced by the edges of $C$ not in $P_1$.

      Clearly, the first two points hold: $P_1 \cap P_2 = \set{u,v}$ and $u,v \in A_1\cap A_2$, and we have seen that $P_1\setminus\set{u,v} \subset A_1\setminus A_2$, and thus $P_1\setminus P_2 \subset A_1\setminus A_2$.
          
      For the third point, by way of contradiction, assume that $P_2\setminus \set{u,v}$ contains a vertex of $w \in A_1\cap A_2$. We have $|P_1| = \dist_{A_1}(u,v) = d(u,v)$. By the choice of $u$ and $v$ on $C$, we must have $|P_1| + |P_2| \ge 2d(u,v)$, i.e., $|P_2| \ge d(u,v)$. On the other hand, let us consider the path $P$ between $u$ and $v$ in $A_2$. We have $|P| = \dist_{A_2}(u,v) = d(u,v)$. So, $P_1 \cup P$ creates a cycle of length at most $2d(u,v)$. By minimality of $C$, this implies that $|P_2| = d(u,v)$. But, if $w\in P_2\setminus \set{u,v}$, then both $\dist_C(u,w)$ and $\dist_C(v,w)$ would be less than $|P_2| = d(u,v) = \dist_C(u,v)$: a contradiction with the choice of $u,v \in A_1\cap A_2$ that minimizes the distance on $C$. So, $P_2\setminus P_1$ has no vertex in $A_1\cap A_2$. Thus, $P_2\setminus P_1$ contains only vertices and edges of a same tree $A_i$. Clearly, $P_2\setminus P_1 \subset A_1$ is not possible, since otherwise $C$ would be wholly contained in $A_1$. So, $P_2\setminus P_1 \subset A_2\setminus A_1$ as required.
  \end{proof}

  Let $u,v,P_1,P_2$ be as in \cref{claim:uv_P1P2}. Note that $d(u,v) > 1$. Indeed, we have seen that the size of $C$ is $|P_1| + |P_2| = 2d(u,v)$, and $d(u,v) = 1$ is not possible since $C$ is a cycle. Then, we can define $w_i$ to be the neighbor of $u$ in $P_i$, for each $i\in\set{1,2}$. See \cref{fig:cyclereduccontrac} on the left. And, consider the graph $\sU'$ obtained from $\sU$ by identifying $w_1$ and $w_2$ into a new vertex $w$, and by removing double edges created. We call this operation on $\sU$ a \emph{contraction}, since it corresponds to a classical edge-contraction if $w_1$ and $w_2$ are adjacent. See \cref{fig:cyclereduccontrac}.

  \begin{myfigure}
  \begin{tabular}{ccc}
      \begin{tikzpicture}[scale=0.76]
      %\draw[help lines] (0,-2) grid (7,+2);
          
      \node at (3,1.2) {\color{mygreen}Tree $A_1$};
      \node at (3,-1.2) {\color{myred}Tree $A_2$};
      \node[mygreen] at (5.5,+1.2) {$P_1$};
      \node[myred] at (5.5,-1.2) {$P_2$};
      \node at (3,0) {$C$};

      % les sommets
        \tikzset{every node/.style={vertex,fill=none, very thick}}
          \node (u) at (0,0) {};
          \node (v) at (6,0) {};

          \tikzset{every node/.append style={draw=mygreen}}    
          \node (w1) at (1,1) {};
          \node (w2) at (2,1.5) {};
          \node (w3) at (5,1) {};

          \tikzset{every node/.append style={draw=myred}}
          \node (ww1) at ($(w1)+(0,-2)$) {};
          \node (ww2) at ($(w2)+(0,-3)$) {};
          \node (ww3) at ($(w3)+(0,-2)$) {};
          
          % le cycle          
          \tikzset{every path/.style={draw=mygreen,semithick}}
          \draw (u) to [bend left=10] (w1) ;
          \draw (w1) to [bend left=10] (w2) ;
          \draw[dashed] (w2) to [bend left=20] (w3);
          \draw (w3) to [bend left=10] (v);
          
          \tikzset{every path/.style={draw=myred,semithick}}
          \draw (u) to [bend right=10] (ww1) ;
          \draw (ww1) to [bend right=10] (ww2) ;
          \draw[dashed] (ww2) to [bend right=20] (ww3);
          \draw (ww3) to [bend right=10] (v);
          
          % les noms
          \tikzset{every node/.style={draw=none,inner sep=3pt}}
          \node at (u) [left=2pt] {$u$};
          \node at (v) [right=2pt] {$v$};
          \node at (w1) [above left] {$w_1$};
          \node at (ww1) [below left] {$w_2$};
      \end{tikzpicture}
      &
      \begin{tikzpicture}[scale=0.76]
          \node at (0,0) {};
          \draw[-Triangle, very thick](0,1.5) -- (1, 1.5);
      \end{tikzpicture}
      &
      \begin{tikzpicture}[scale=0.76]
      %\draw[help lines] (0,-2) grid (15,+2);
      
      \node at (3,1.2) {\color{mygreen}$A_1$};
      \node at (3,-1.2) {\color{myred}$A_2$};
      
      % Les sommets
      \tikzset{every node/.style={vertex,fill=none,very thick}}
      \node (u) at (0,0) {};
      \node (v) at (6,0) {};
      \node (w1) at (1,0) {};

      \tikzset{every node/.append style={draw=mygreen}}
      \node (w2) at (2,1.5) {};
      \node (w3) at (5,1) {};
      
      \tikzset{every node/.append style={draw=myred}}        
      \node (ww2) at ($(w2)+(0,-3)$) {};
      \node (ww3) at ($(w3)+(0,-2)$) {};
          
      % Le cycle          
      \tikzset{every path/.style={draw=mygreen,semithick}}
      \draw (w1) to [bend left=12] (w2) ;
      \draw[dashed] (w2) to [bend left=20] (w3);
      \draw (w3) to [bend left=10] (v);
          
      \tikzset{every path/.style={draw=myred,semithick}}
      \draw (w1) to [bend right=12] (ww2) ;
      \draw[dashed] (ww2) to [bend right=20] (ww3);
      \draw (ww3) to [bend right=10] (v);

      \tikzset{every path/.style={thick}}
      \draw (u) to (w1);
          
      % Les noms
      \tikzset{every node/.style={draw=none,inner sep=3pt}}
      \node at (u) [left=2pt] {$u$};
      \node at (v) [right=2pt] {$v$};
      \node at (w1) [right=2pt] {$w$};
      \node at (w1) [above left, color=gray] {$w_1$};
      \node at (w1) [below left, color=gray] {$w_2$};
      \end{tikzpicture}
      \end{tabular}
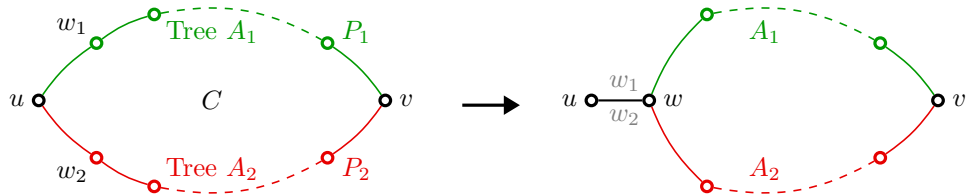
      \caption{Cycle $C = P_1 \cup P_2$. Contracting $w_1$ and $w_2$ in $\sU$. Green vertices and edges are in $A_1\setminus A_2$, red vertices and edges are in $A_2\setminus A_1$, and black vertices and edges are in $A_1\cap A_2$.}
      \label{fig:cyclereduccontrac}
  \end{myfigure}

  We will show that $\sU'$ is still isometric-universal for $\sF$, contradicting the fact that $\sU$ has the minimum number of vertices. This will prove that $C$ does not exist in $\sU$, and so $\sU$ is a tree.

  Let us show that $A_1$ and $A_2$ are isometric subgraphs of $\sU'$. We have to prove that the contraction of $w_1$ and $w_2$ does not change any distance between any pair of vertices that belong to a same tree $A_i$, for each $i \in\set{1,2}$.

  First we remark that $A_i$ is a subgraph of $\sU'$. Indeed, $w_1$ and $w_2$ belong to different trees, and thus the edge $u-w_i$ of $A_i$ in $\sU$, after contraction, is in $A_1\cap A_2$ in $\sU'$. So, the distances between vertices taken in $A_i$ can only decrease in $\sU'$. 

  We proceed by way of contradiction. Assume there are two vertices $x$ and $y$ at distance $d$ in $\sU$ that both belong to $A_i$, for some $i$, and such that $\dist_{\sU'}(x,y) < d$. We consider only the case $A_i = A_1$, the case $A_i = A_2$ being symmetric.

  Let us show that:

  \begin{claim}\label{claim:P'}
  There exists a path $P$ between $x$ and $y$ in $\sU$ of length $|P| \le d+1$ that passes through the path $w_1 - u - w_2$ or the path $w_2 - u - w_1$. 
  \end{claim}

  \begin{proof}
  For every path $Q$ in $\sU$, we denote by $\hQ$, its \emph{contraction}, i.e., the path obtained by contracting all the edges of the subpath of $Q$ between $w_1$ and $w_2$. If $w_1$ or $w_2$ is not in $Q$, we set $\hQ = Q$. Note that $\hQ$ is a path of $\sU'$.

  Let $P'$ be a shortest path in $\sU'$ between $x$ and $y$. By hypothesis on $x,y$, $|\hP| < d$. Let us show that there exists a path $P$ in $\sU$ such that $P' = \hP$ and $|P| \le d+1$. If $P'$ does not contain $w$, then $P = P' = \hP$. Otherwise, $w \in P'$. There are two cases: (1) if $P'$ contain $uw$, then we can set $P$ as the path $P'$ where $w$ is replaced $w_1$ or $w_2$ (depending on the other edge incident to $w$ in $P'$); (2) if $P'$ does not contain $uw$, then since $P'$ is a shortest path, it cannot contain $u$. We can set $P$ as the path $P'$ where $w$ is replaced by the path $w_1 - u - w_2$ or $w_2 - u - w_1$. We check that in both cases, $\hP = P'$, and $|P| \le |\hQ| + 2 < d + 2$, that implies $|P| \le d+1$.
  %By hypothesis on $x,y$, $|\hP| < d$. If $w_1$ or $w_2$ is not in $P$, then $|\hP| = |P| \ge d$, since $P$ connects $x,y$ in $\sU$. Thus, $P$ must contain $w_1$ and $w_2$.
  %
  %Now, the path $P$ is obtained from $Q$ by replacing its subpath between $w_1$ and $w_2$ by the path $w_1 - u - w_2$ (or $w_2 - u - w_1$, depending in which order $w_1,w_2$ are traversed by $Q$). 
  %
  %We have $|P| = |\hQ| + 2 < d + 2$, that implies $|P| \le d+1$.
  \end{proof}

  Let $P$ be a path as in \cref{claim:P'}. Note that $P$ is a witness for $\dist_{\sU'}(x,y) < d$, i.e., a path connecting $x$ to $y$ and such that after contracting $w_1$ and $w_2$, its length decreases by $2$ in $\sU'$, and becomes less than $d$. W.l.o.g. we will assume that $P$ visits $w_1 - u - w_2$ in this order when going from $x$ to $y$. If not, we can exchange the role of $x$ and $y$.

  Now, let us define the vertices $a,b,z$ as follows (see \cref{fig:cyclereduccontracnot}):

  \begin{itemize}

  \item Vertex $b$ is the last vertex of $P$, when going from $u$ to $y$, that belongs to $P_2$. Note that $b$ is between $w_2$ and $y$ on $P$, possibly with $b = w_2$. See \cref{fig:cyclereduccontracnot} on the left.

  \item Vertex $z$ is the first vertex of $P$, when going from $b$ to $y$, that belongs to $A_1$, possibly with $z = y$. Note that $z \neq b$. See \cref{fig:cyclereduccontracnot} on the right.

  \item Vertex $a$ is the last vertex of the path between $u$ to $z$ in $A_1$, when going from $u$ to $z$, that belongs to $P_1$, possibly with $a \in\set{u,w_1,v}$. See \cref{fig:cyclereduccontracnot} on the right.

  \end{itemize}

  \begin{myfigure}
  \begin{tabular}{cc}
      \begin{tikzpicture}[scale=0.76]
      %\draw[help lines] (0,-2) grid (8,+2);
      
      \node[text width=3cm] at (4,1.2) {\color{mygreen}$A_1$};
      \node[text width=3cm] at (4,-1.2) {\color{myred}$A_2$};
      \node[thick,rounded corners=.1cm, fill=mygray] at (2,2.3) {$P$};

      % Les sommets en "hilight"
      \tikzset{every node/.style={draw=mygray,circle,fill=mygray,inner xsep=2\pgflinewidth}}
      \node at (0,0) {}; % u
      \node at (1,-2) {}; % z
      \node at (4.5,2.5) {}; % x
      \node at (0,-2.3) {}; % y
      \node at (1,1) {}; % w1
      \node at ($(1,1)+(0,-2)$) {}; % w2
      \node at (3.5,1.7) {}; % a
      \node at (3.5,-1.7) {}; % b

      % Les sommets
      \tikzset{every node/.style={draw=black,vertex,fill=none,very thick}}
      \node (u) at (0,0) {};
      \node (v) at (6,0) {};
      \node (z) at (1,-2) {};

      \tikzset{every node/.append style={draw=mygreen}}
      \node (w1) at (1,1) {};
      \node (w1p) at (5,1) {};
      \node (a) at (3.5,1.7) {};
      \node (x) at (4.5,2.5) {};
      \node (y) at (0,-2.3) {};
          
      \tikzset{every node/.append style={draw=myred}}
      \node (w2) at ($(w1)+(0,-2)$) {};
      \node (w2p) at ($(w1p)+(0,-2)$) {};
      \node (b) at ($(a)+(0,-3.4)$) {};

      % Le chemin P en "hilight"
      \tikzset{every path/.style={draw=mygray,ultra thick,-latex,double=mygray,double distance=2\pgflinewidth,-}}
      \draw (u)  to [bend left=10]  (w1) ;
      \draw (w1) to [bend left=20]  (a);
      \draw (a)  to [bend left=10]  (x);
      \draw (z)  to [bend left=10]  (y);
      \draw (u)  to [bend right=10] (w2) ;
      \draw (w2) to [bend right=20] (b) ;
      \draw (b)  to [bend left=10]  (z);

      % Le chemin P
      \tikzset{every path/.style={draw=mygreen,semithick}}
      \draw[dashed] (a) to [bend left=10] (x);
      \draw[dashed] (w1) to [bend left=20] (a);
      \draw (u) to [bend left=10] (w1) ;
      \draw[dashed](z) to [bend left=10] (y);
      
      \tikzset{every path/.style={draw=myred,semithick}}
      \draw (u) to [bend right=10] (w2) ;
      \draw[dashed] (w2) to [bend right=20] (b) ;
      \draw[dashed] (b) to [bend left=10] (z);
              
      % Le cycle
      \tikzset{every path/.style={draw=mygreen,semithick}}
      \draw[dashed] (a) to [bend left=10] (w1p);
      \draw (w1p) to [bend left=10] (v);
          
      \tikzset{every path/.style={draw=myred,semithick}}
      \draw[dashed] (b) to [bend right=10] (w2p);
      \draw (w2p) to [bend right=10] (v);

      % Les noms
      \tikzset{every node/.style={draw=none}}
      \node at (u) [left=2pt] {$u$};
      \node at (v) [right=2pt] {$v$};
      \node at (w1) [above left] {$w_1$};
      \node at (w2) [below left] {$w_2$};
      \node at (x) [right=2pt] {$x$};
      \node at (b) [below right] {$b$};
      \node at (y) [left=2pt] {$y$};
      \end{tikzpicture}
      &
      \begin{tikzpicture}[scale=0.76]
      %\draw[help lines] (0,-2) grid (6,+2);
      
      \node[text width=3cm] at (4.5,1.2) {\color{mygreen}$A_1$};
      \node[text width=3cm] at (4,-1.2) {\color{myred}$A_2$};
      
      % Les sommets
      \tikzset{every node/.style={vertex,fill=none,very thick}}
      \node (u) at (0,0) {};
      \node (v) at (6,0) {};
      \node (z) at (1.8,-2) {};

      \tikzset{every node/.append style={draw=mygreen}}
      \node (w1) at (1,1) {};
      \node (w1p) at (5,1) {};
      \node (a) at (2.5,1.7) {};
      \node (y) at (0,-2.3) {};
      \node (zm) at (0.9,-2.1) {};
          
      \tikzset{every node/.append style={draw=myred}}
      \node (w2) at ($(w1)+(0,-2)$) {};
      \node (w2p) at ($(w1p)+(0,-2)$) {};
      \node (b) at (3.5,-1.7) {};
          
      % Le cycle          
      \tikzset{every path/.style={draw=mygreen,semithick}}
      \draw (u) to [bend left=10] (w1) ;
      \draw[dashed] (w1) to [bend left=10] (a);
      \draw[dashed] (a) to [bend left=20] (w1p);
      \draw (w1p) to [bend left=10] (v);
      \draw[dashed](z) to [bend left=10] (zm);
      \draw[dashed](zm) to [bend left=10] (y);
      \draw[dashed](a) to [bend right=80, looseness=2.3] (zm);
         
      \tikzset{every path/.style={draw=myred,semithick}}
      \draw (u) to [bend right=10] (w2) ;
      \draw[dashed] (w2) to [bend right=20] (b) ;
      \draw[dashed] (b) to [bend right=10] (w2p);
      \draw (w2p) to [bend right=10] (v);
      \draw[dashed] (b) to [bend left=10] (z);

      % Les noms
      \tikzset{every node/.style={draw=none}}
      \node at (u) [left=2pt] {$u$};
      \node at (v) [right=2pt] {$v$};
      \node at (w1) [above left] {$w_1$};
      \node at (w2) [below left] {$w_2$};
      \node at (b) [below right] {$b$};
      \node at (y) [left=2pt] {$y$};
      \node at (a) [above right] {$a$};
      \node at (z) [below right] {$z$};
      \end{tikzpicture}
      \end{tabular}
      \caption{Path $P$ between $x$ and $y$, and vertices $a,b,z$.}
      \label{fig:cyclereduccontracnot}
  \end{myfigure}

  Since $P$ passes successively through vertices $x, w_1, u, w_2, b, z, y$, we have:
  \[
      d(x,w_1) + d(w_1,u) + d(u,w_2) + d(w_2,b) + d(b,z) + d(z,y) ~\le~ |P| ~.
  \]
  The quantity $d(w_1,u) + d(u,w_2)$ is equal to $2$, since $w_1,u,w_2$ are consecutive in $P$. On the other hand, $|P| \le d + 1$ by definition of $P$. By the triangle inequality between $x$ and $y$, $d = d(x,y) \le d(x,w_1) + d(w_1,a) + d(a,z) + d(z,y)$. Therefore:
  \begin{align}
   d(x,w_1) + 2 + d(w_2, b) + d(b,z) + d(z,y) &\le d(x,w_1) + d(w_1,a) + d(a, z) + d(z,y) + 1 \nonumber\\
  \Rightarrow\quad 1 + d(w_2, b) + d(b,z) &\le d(w_1,a) + d(a, z) \label{eq:abz}
  \end{align}

  In fact, we will show that:
  \begin{claim}\label{claim:abz}
      $d(w_1,a) + d(a, z) \le d(w_2, b) + d(b,z)$.
  \end{claim}

  Before presenting its proof, observe that the inequality of \cref{claim:abz} contradicts \cref{eq:abz}, and the fact that some distance between two vertices of the same tree, namely $x$ and $y$, is shortened in $\sU'$. So $\sU'$ is universal-isometric for $\set{A_1,A_2}$, and thus $\sU$ cannot be minimum, and thus does not contain any cycle.

  \begin{proof}%[Proof of \cref{claim:abz}]
  We will actually prove a stronger statement which is:
  \[
      d(w_1,a) = d(w_2, b) \quad\mbox{and}\quad d(a, z) = d(b,z) ~.
  \]

  We first remark that vertex $z\in A_1\cap A_2$. Indeed, we have $z\in A_1$. If $z\notin A_2$, then the edge $z-z'$ on the path $P$ from $z$ to $b$ is not in $A_2$. So, $z-z' \in A_1$, which implies $z'\in A_1$: a contradiction with the fact that $z$ is the first vertex on $P$ from $b$ to $y$ that is in $P_1$. So $z \in A_1\cap A_2$.

  As a preliminary step, let us show that $d(u,a) =  d(u,b)$, cf. \cref{eq:ua_ub} hereafter. Recall that $u,v,z \in A_1 \cap A_2$, and that distances in $A_i$ coincide with distances in $\sU$.

  By definition, $a$ is on the path from $u$ to $z$ in $A_1$. Thus
  \begin{align*}
      \dist_{A_1}(u, z) &= \dist_{A_1}(u, a) + \dist_{A_1}(a, z) \\
      \Rightarrow\quad d(u, z) &= d(u, a) + d(a, z)
  \end{align*}
  Similarly for $b$, which is on the path in $A_2$ from $u$ to $z$, we have $d(u, z) = d(u, b) + d(b, z)$. Combining these two equations, we get:
  \begin{align}
      d(u, a) + d(a, z) = d(u, b) + d(b, z) \label{eq:zaau_zbbu}
  \end{align}

  We use the following fact:

  \begin{fact}%  \textbf{Fact.}
  For any three vertices in a tree, the three paths between them have
  exactly one vertex in common.
  \end{fact}

  In particular, whenever applied to $u,v$ and $z$ in $A_1$, pairwise connected by the three tree paths $P_{uv}, P_{vz}, P_{uz}$ in $A_1$, the fact implies the last vertex of $P_{uz}$ that is on $P_{uv}$ is also the last vertex of $P_{vz}$ that is on $P_{uv}$. From this, we deduce that $a$ is on the path from $v$ to $z$ in $A_1$. Thus $d(v, z) = d(v, a) + d(a, z)$. With the same argument applied to the same vertices but in $A_2$, we deduce that $b$ is on the shortest path in $A_2$ from $v$ to $z$, we have $d(v, z) = d(v, b) + d(b, z)$. Therefore:
  \begin{align}
      d(v, a) + d(a, z) = d(v, b) + d(b, z)
    \label{eq:zaav_zbbv}
  \end{align}
  By subtracting \cref{eq:zaav_zbbv} from \cref{eq:zaau_zbbu}, we have:
  \begin{align}
          && (d(u, a) + d(a, z)) - (d(v, a) + d(a, z)) &= (d(u, b) + d(b, z)) - (d(v, b) + d(b, z)) \nonumber\\
          \Rightarrow&& d(u, a) - d(v, a) &= d(u, b) - d(v, b) \label{eq:ua=ub}
  \end{align}
  Note that $P_1$ and $P_2$ are paths of $A_1$ and $A_2$ respectively, thus they are shortest paths in $\sU$. Since $a$ and $b$ are on $P_1$ and $P_2$ respectively, we have $d(u,v) = d(u,a) + d(a,v) = d(u,b) + d(b,v)$. And by adding to \cref{eq:ua=ub}, we get:
  \begin{align}
      d(a,u) = d(b,u) \label{eq:ua_ub}
  \end{align}
  as claimed.

  We now show that:
  \begin{align}   
      d(a,w_1) = d(b,w_2) \label{eq:w1a_w2b}
  \end{align}
  Indeed, either $d(a,u) = d(b,u) = 0$, i.e., $a = u = b$, and thus $d(a,w_1) = d(u,w_1) = 1 = d(u,w_2) = d(b,w_2)$. Or, $d(a,u) = d(b,u) > 0$, and thus $w_1$ is on the path from $u$ to $a$ on $P_1$, and $w_2$ is on the path from $u$ to $b$ on $P_2$. This implies $d(u,w_1) + d(w_1,a) = d(u,w_2) + d(w_2,b)$, and thus $d(w_1,a) = d(w_2,b)$.

  Furthermore, applying \cref{eq:ua_ub} on \cref{eq:zaau_zbbu} gives that:
  \begin{equation}
      d(a,z) = d(b,z) \label{eq:az_bz}
  \end{equation}
  \cref{eq:w1a_w2b} and \cref{eq:az_bz} completes the proof of \cref{claim:abz}.
  \end{proof}
  This completes the proof of \cref{th:2trees}.
  \end{proof}

  This structural result allows us to predict that every minimum isometric-universal graph of two trees $\set{T_1,T_2}$ either is a tree, or contains an isometric-universal graph for $\set{T_1,T_2}$ as a subgraph (only by removing edges). However, as the proof of
  \cref{th:2treesalgo} will explain, if we can compute a tree that contains $T_1$ and $T_2$ as subgraphs, then we can compute a minimum isometric-universal graph for $\set{T_1,T_2}$. More precisely:

  \begin{theorem}\label{th:2treesalgo}
      A minimum isometric-universal graph for a pair of trees with at most $n$ vertices each can be computed in time $O(n^{5/2}\log{n})$.
  \end{theorem}

  \begin{proof}
      Let $T_1, T_2$ be two trees with at most $n$ vertices. From \cref{th:2trees}, there always exists a minimum isometric-universal graph of two trees that is a tree itself. Thus, any algorithm that computes a minimum isometric-universal tree for $\set{T_1, T_2}$ computes a minimum isometric-universal graph for $\set{T_1, T_2}$ too. 

      In a tree there exists only one path connecting any pair of vertices. Thus, every tree that contains $T_1$ and $T_2$ as subgraphs (such a tree is called a \emph{supertree} under graph isomorphism in \cite{GN98}) is also an isometric-universal tree. According to Gupta and Nishimura~\cite[Theorem~9.6]{GN98}, a minimum supertree of two trees can be computed in time  $O(n^{5/2}\log{n})$.
      
      From the above discussion, the minimum supertree for $T_1,T_2$ is a minimum isometric-universal tree and also a minimum isometric-universal graph for $\set{T_1,T_2}$, which concludes the proof.
   \end{proof}

%------------------------------------------
\subsection{The case of two forests}

Now we extend the previous result to forests instead of trees. The combination of \cref{th:2graphs} and \cref{th:2treesalgo} leads to the following result:

\begin{theorem}\label{th:2forests}
    A minimum isometric-universal graph for a pair of forests with at most $n$ vertices each can be computed in time $O(n^{7/2}\log{n})$.
\end{theorem}

\begin{proof}
    This is a direct application of \cref{th:2graphs} applied on two forests, composed respectively of $s$ and $r$ trees. In the complexity formula of \cref{th:2graphs}, we can bound $r$, $s$, and $f(n)$ as follows:
    \begin{itemize}%[noitemsep]
        \item $s \le r \le n$.
        % Cyril: il faut mettr la base du log car on a pas de O(...) ici.
        \item $f(n) = C n^{5/2}\log_2{n}$, for some constant $C>0$, since, from \cref{th:2treesalgo}, the time complexity of finding a minimum isometric-universal graph of two trees is at most $f(n)$. Note that $f$ is superlinear.
    \end{itemize} 

   We have:\\

   $r\min\set{f(2n),\, s f(n)} + s \min\set{ r \sqrt{s}\log{n},\, rs + s\log\log{s}}$
   \vspace{-5pt}
    \begin{eqnarray*}
             &\le& rf(2n) + sr \sqrt{s}\log{n}\\
             &=& O(n^{7/2}\log{n} + n^{5/2}\log{n})\\
             &=& O(n^{7/2}\log{n}) ~.
    \end{eqnarray*}
    Applying \cref{th:2graphs}, this completes the proof.
\end{proof}

%%%%%%%%%%%%%%%%%%%%%%%%%%%%%%%%%%%%%%%%%%%%%%%%%%%%%%%%%%%%%%%%%%%%%% 

\section{Universal Graph for Three Forests}
\label{sec:3forests}

%\Cyril{Une clarification s'impose, car dans \cite[p.~209]{GN98}, il est écrit: "Keselman and Amir have shown that the three-input version of the Largest Common Embeddable Tree Problem is NP-complete for subgraph isomorphism \cite{AK97b}. There are no known such results for the SCES Problem." Est-ce que "Largest Common Embeddable Tree" est pareil que "LCS"? Comme LCS est pareil que SCE, du coup, on aurait pour le résultat pour SCE et répondrait à une question de \cite{GN98}. En tout cas, \cite{AK97b} est basé sur une simple réduction à 3DM qui ne marche pas pour nous. Malheureusement, il va falloir expliquer que cela n'est pas pareil ... En gros, nous les feuilles ne sont pas étiquetées.}

%------------------------------------------
\subsection{Three trees and more} %\cref{th:2trees} to more trees}

\Status{Edgar: Passe faite jusqu'à la fin de la preuve de \cref{th:3trees}}

Unfortunately, \cref{th:2trees} cannot be extended to three trees. Indeed, as we will show in \cref{th:3trees}, there are triples of trees with no minimum isometric-universal graph that is a tree. As we will see later in \cref{subsec:greedy}, this implies that no greedy strategy can build a minimum isometric-universal graph, even if equipped with an optimal oracle for isometric-universal graphs for a pair of graphs.

\begin{theorem}\label{th:3trees}
  For every integer $t\ge 9$, there is a family of trees $\set{T_1,T_2,T_3}$, each with $2 t^3 + 2$ vertices and with pathwidth\footnote{The pathwidth of a graph $G$ is the smallest integer $p$ such that $G$ is a subgraph of some interval graphs with maximum clique size $p+1$.} two, for which every isometric-universal tree is not minimum.
\end{theorem}

\begin{proof}
    Each tree of $\sF = \set{T_1,T_2,T_3}$ is composed of two disjoint stars uniformly subdivided with an edge between their centers. These two stars are chosen among a set of three different subdivided stars.
    
    More precisely, a \emph{type-$i$ star}, for $i\in\set{1,2,3}$, is the graph obtained from a $K_{1,t^i}$ where each edge is replaced by a path of length $t^{3-i}$. For convenience, we let $s = t^3 + 1$, so that every type-$i$ star has precisely $s$ vertices. The \emph{center} of a type-$i$ star is its unique vertex of degree $t^i > 2$, assuming $t\ge 3$.

    %\Cyril{Au lieu de nommer les stars $S_i$, je leur donne un type. Cela paraît bizarre, mais dans la suite on va vouloir dire des choses en fonction du "type" des étoiles mais aussi des étoiles elles-mêmes. Il va donc avoir des mélanges entre le type des étoiles et une étoile particulière d'un $T_i$. En fait $T_i$ et $T_j$ ont un type d'étoile en commun plutôt qu'une étoile $S$ en commun. Cela peut devenir obscure dans $\hT_i$ et $\hT_j$ puisque maintenant ces sous-graphes s'intersectent ...}

    Each $T_i \in \sF$, for $i\in\set{1,2,3}$, is obtained by taking a copy of a type-$i$ star and a copy of a type-$(1+(i \bmod 3))$ star, and by linking their centers by an edge. See \cref{fig:familycycletree}. We have $\cardV*{T_i} = 2s = 2t^3 + 2$ as required. We check that $T_i$ is a tree of pathwidth two. \label{page:pw2} (The pathwidth upper bound comes from the fact that each star without its center is a union of paths (a.k.a. a linear forest), and so of pathwidth one, and lower bound from the fact that $T_i$ contains the $Y$-graph, an obstruction for pathwidth-1 graphs.)

    % \ETOILE{X}{type(=1,2,3)}{deg(=t,t^2,t^3)}{length(=t^2,t,1)}{COLOR}
    \newcommand{\ETOILE}[5]{%
    \tikzset{every node/.style={vertex,fill=none}}
    
        % L'étoile
      \node (c#2) at (#1,0) {};
    
      \node (l1#2) at ($(c#2)+(-0.75,1)$) {};
      \node (l2#2) at ($(c#2)+(-0.5,1)$) {};
      \node (l3#2) at ($(c#2)+(+0.75,1)$) {};
    
      \node (l1u1#2) at ($(l1#2)+(0,0.4)$) {};
      \node (l2u1#2) at ($(l2#2)+(0,0.4)$) {};
      \node (l3u1#2) at ($(l3#2)+(0,0.4)$) {};
    
      \node (l1u3#2) at ($(l1u1#2)+(0,0.8)$) {};
      \node (l2u3#2) at ($(l2u1#2)+(0,0.8)$) {};
      \node (l3u3#2) at ($(l3u1#2)+(0,0.8)$) {};
    
      \tikzset{every node/.style={draw=none,inner sep=3pt}}
     
      % Arêtes de l'étoile
      \draw (c#2) -- (l1#2);
      \draw (c#2) -- (l2#2);
      \draw (c#2) -- (l3#2);
    
      \draw (l1#2) -- (l1u1#2);
      \draw (l2#2) -- (l2u1#2);
      \draw (l3#2) -- (l3u1#2);
    
      \draw[dotted] (l1u1#2) -- (l1u3#2);
      \draw[dotted] (l2u1#2) -- (l2u3#2);
      \draw[dotted] (l3u1#2) -- (l3u3#2);
    
      \draw[dotted] (-0.2+#1,1.8+0) -- (0.5+#1,1.8+0);
    
        % accolades
        %\tikzset{every path/.style={accolade}}
      \draw[dashed, <->] ($(l1#2)+(-0.1,1.6)$) -- ($(l3#2)+(0.1,1.6)$) node[midway,above] {$#3$};
      \draw[dashed, <->] ($(l1#2)+(-0.4,-0.1)$) -- ($(l1u3#2)+(-0.4,+0.1)$) node[midway,left] {$#4$};
    
      % Rectangle coloré
      \draw[draw=#5, rounded corners, opacity=.9] (-1.9+#1,-0.35+0) rectangle (1+#1,3.7);
      \node at (-1+#1,3.4) [text=#5] {type-$#2$};
    }
    
    \begin{myfigure}
        \begin{tikzpicture}[scale=0.68]
        \ETOILE{0}{1}{t}{t^2}{RedOrange}
        \ETOILE{3.35}{2}{t^2}{t}{DarkGreen}
        \draw (c1) -- (c2);
        \end{tikzpicture}
        \hfill
        \begin{tikzpicture}[scale=0.68]
        \ETOILE{0}{2}{t^2}{t}{DarkGreen}
        \ETOILE{3.35}{3}{t^3}{1}{RoyalBlue}
        \draw (c2) -- (c3);
        \end{tikzpicture}
        \hfill
        \begin{tikzpicture}[scale=0.68]
        \ETOILE{0}{3}{t^3}{1}{RoyalBlue}
        \ETOILE{3.35}{1}{t}{t^2}{RedOrange}
        \draw (c3) -- (c1);
        \end{tikzpicture}
        \\
        \hspace{21mm} $T_1$ \hfill $T_2$ \hfill $T_3$ \hspace{21mm}
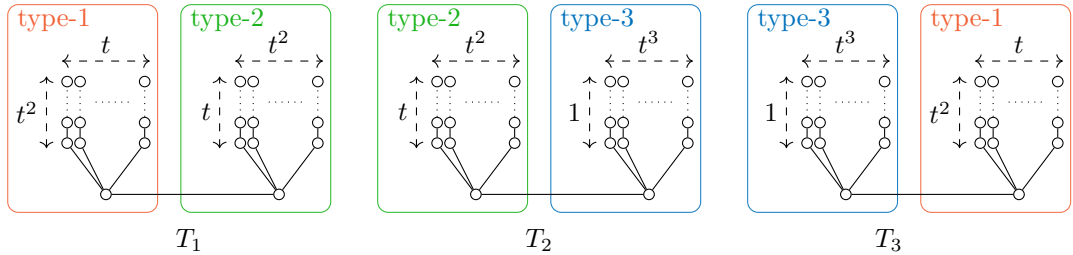
    \caption{A tree family $\set{T_1,T_2,T_3}$ whose every minimum isometric-universal graph is not a tree.}  
    \label{fig:familycycletree}
    \end{myfigure}

    Let $\sU$ (resp. $\sT$) be any minimum isometric-universal graph (resp. tree) for $\sF$. Our goal is to show that $\cardV*{\sT} > \cardV*{\sU}$. 

    Clearly, $\sF$ has a cyclic isometric-universal graph with $3s$ vertices constructed as follows (cf. \cref{fig:univfamilycycletree}). It is composed of disjoint type-$1$, type-$2$, and type-$3$ stars where a cycle between their centers is added. 
    
    \begin{myfigure}
      \begin{tikzpicture}[scale=0.9]
        \ETOILE{0}{1}{t}{t^2}{RedOrange}
        \ETOILE{4}{2}{t^2}{t}{DarkGreen}
        \ETOILE{8}{3}{t^3}{1}{RoyalBlue}
        \draw (c1) -- (c2);
        \draw (c2) -- (c3);
        \draw (c1) to [bend right=20]  (c3);
      \end{tikzpicture}
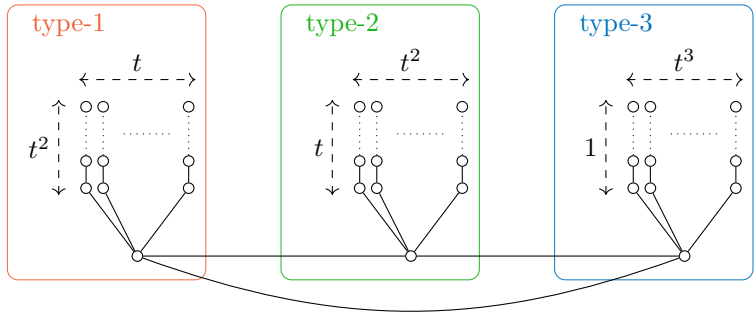
      \caption{A cyclic isometric-universal graph for $\sF$ that is smaller than any isometric-universal tree for $\sF$.}
        \label{fig:univfamilycycletree}
    \end{myfigure}

    It is easy to check that this graph is isometric-universal for $\sF$, and has $3s$ vertices. Therefore, we have the upper bound $3s \ge \cardV*{\sU}$. The remainder of the proof is to show a lower bound for $\sT$, that is $\cardV*{\sT} > 3s$. 
    
    Consider any isometric embedding of $T_1,T_2,T_3$ in $\sT$, and denote by $\hT_1,\hT_2,\hT_3$ their corresponding isometric subgraphs in $\sT$. From the inclusion-exclusion principle applied to the number of vertices of $\hT_1\cup \hT_2\cup \hT_3$, we have for any range $I$ of indices:
    \[
        \card{ V\pare{ \bigcup_{i\in I} \hT_i}} = \sum_{J\subseteq I} (-1)^{|J|+1} \card{ V\pare{ \bigcap_{j\in J} \hT_j}}
    \]
    which implies for $I = \set{1,2,3}$:
    \[
        \cardV{\hT_1 \cup \hT_2 \cup \hT_3} = \sum_{i=1}^3 \cardV{\hT_i} - \sum_{i\neq j} \cardV{ \hT_i \cap \hT_j } + \cardV{\hT_1\cap \hT_2\cap \hT_3}
    \]
    Hence, we have:
    \begin{eqnarray*}        
        \cardV{\sT} &\ge& \cardV{\hT_1 \cup \hT_2 \cup \hT_3} \ge \sum_{i=1}^3 \cardV{\hT_i} - \sum_{i\neq j} \cardV{\hT_i \cap \hT_j}\\
        &\ge& 6s - \sum_{i\neq j} \cardV{\hT_i \cap \hT_j}
    \end{eqnarray*}

    So, it suffices to show that the term $\sum_{i\neq j} \cardV*{\hT_i \cap \hT_j} < 3s$, i.e., the pairwise intersection of trees is not too large. In fact, we will show that $\sum_{i\neq j} \cardV*{\hT_i \cap \hT_j} \le 2s + O(s^{2/3})$, which suffices for $s$ large enough.
    
    In order to analyze $\cardV*{\hT_i \cap T_j}$, we first observe that stars of different types intersect poorly in $\sT$, and similarly for stars with different centers. More precisely:
    
    \begin{claim}\label{claim:inter_S}
        Let $S,S'$ be two type-$p$ and type-$p'$ stars with centers $c$ and $c'$ that are isometric subgraphs of $\sT$. Then, if $p\leq p'$,
        \[
            \cardV{S \cap S'} \le \left\{%
            \begin{array}{ll}
                 1 + 2t^{3-p'} & \mbox{if $c \neq c'$}\\
                 1 + t^{p+3-p'} & \mbox{if $p \neq p'$}\\
            \end{array}
            \right.
        \]
    \end{claim}

    \begin{proof}
        First, let us show that, if $c \neq c'$, then $\cardV*{S\cap S'} \le 1 + 2 t^{3-p'}$.
        
        Every vertex, apart from the center, has degree at most two in $S$ or in $S'$, and thus in $S\cap S'$. So, if $c \neq c'$, all vertices of $S\cap S'$ are of degree at most two, and thus $S\cap S'$ must be a linear forest. This forest has only one component. Indeed, if $P,Q$ are two components of $S\cap S'$, then there must be in $S$ a path connecting $P$ to $Q$, and another one in $S'$ connecting $P$ to $Q$. These paths must be disjoint (otherwise they would be in the intersection), and thus form, possibly with the help of $P,Q$, a cycle in $\sT$: a contradiction. Therefore, the quantity $|E(S\cap S')|$ is upper bounded by the minimum of the longest path in $S$ and $S'$, which is $2\min(t^{3-p}, t^{3-p'}) = 2 t^{3-p'}$. It follows that $\cardV*{S\cap S'} \le 1 + 2t^{3-p'}$ as claimed.

        Second, let us show that, if $p \neq p'$ (i.e., $p<p'$), then $\cardV*{S\cap S'} \le 1 + t^{p+3-p'}$.

        If $c \neq c'$, then the previous case applies, and $\cardV*{S\cap S'} \le 1 + 2t^{3-p'} \le 1 + t^{p+3-p'}$, since $t^p \ge 2$, and we are done.

        So, assume that $c = c'$. Vertex $c$ has at most $\min(t^p,t^{p'}) = t^p$ neighbors in $S\cap S'$. Moreover, every path of $S$ or $S'$ starting from $c$ and that does not contain a neighbor in $S\cap S'$ does not contain any other vertex of $S\cap S'$ (it creates a cycle otherwise). Thus, $S\cap S'$ contains at most $t^p$ paths starting from $c$. On the other hand, all vertices of $S\cap S'$ must be at distance at most $\min(t^{3-p}, t^{3-p'}) = t^{3-p'}$ from $c$, because $S'$ has no farther vertices. It follows that $\cardV*{S\cap S'} \le 1 + t^p \cdot t^{3-p'} = 1 + t^{p+3-p'}$ as claimed.
        \end{proof}

    By construction, each tree $T_i$ is the union of two stars of different types (with a single edge between their centers), and each pair $\set{T_i,T_j}$ of trees have exactly one type-$k$ star in common, for some $k$. Actually, $k = 2i - j + 2$ for $i<j$. For instance, $T_1$ and $T_2$ have a type-$2$ star in common, which also appears as subgraph in $\hT_i$ and in $\hT_j$. However, having a common type-$k$ star for the pair $\set{T_i,T_j}$ does not imply that the subgraph $\hT_i \cap \hT_j$ contains a type-$k$ star. Obviously, the way the copies of the type-$k$ stars of $\hT_i$ and of $\hT_j$ intersect will play an important role for $\cardV*{\hT_i\cap \hT_j}$.
    
    Consider a pair $\set*{\hT_i,\hT_j}$ such that $\set{T_i,T_j}$ have a type-$k$ star in common, and let $S_i,S_j$ be the type-$k$ star respectively in $\hT_i$ and in $\hT_j$. Then, the pair $\set*{\hT_i,\hT_j}$ is called \emph{centric} if the centers of $S_i$ and $S_j$ coincide in $\sT$.
    
    The important observation, because $\sT$ is acyclic, is:

    \begin{claim}\label{claim:centric}
        There is a pair $\set*{\hT_{i_0},\hT_{j_0}}$ that is not centric.
    \end{claim}

    \begin{proof}
        By way of contradiction, assume the three pairs $\set*{\hT_1,\hT_2}$, $\set*{\hT_2,\hT_3}$, $\set*{\hT_3,\hT_1}$ are centric. So, by definition, for each pair $\set*{\hT_i,\hT_j}$ there is some vertex in $\sT$ corresponding to the center shared by the common type-$k$ star of $\hT_i$ and of $\hT_j$. Let $c_k$ be this center.
        
        Tree $\hT_2$ contains $c_2$ (because $\set*{T_1,T_2}$ have a type-$2$ star in common), and also $c_3$ (because $\set*{T_2,T_3}$ have a type-$3$ in common). In particular, $c_2$ and $c_3$ are at distance 1 in $\hT_2$. Similarly, $\hT_3$ contains $c_3$ and $c_1$ that are at distance 1 in $\hT_3$. And, similarly, $\hT_1$ contains $c_1$ and $c_2$ that are at distance 1 in $\sT_1$.
        
        It follows that $c_1,c_2,c_3$ induce a cycle in $\hT_1 \cup \hT_2 \cup T_3 \subseteq \sT$: a contradiction.
    \end{proof}
    
    We are now ready to upper bound the number of vertices in any intersecting pair of trees.

    \begin{claim}\label{claim:inter_T}
        For all $i\neq j$, %\Edgar{Un peu trop de $t$ dans les formules}
        \[
            \cardV{\hT_i \cap \hT_j} \le \left\{%
            \begin{array}{ll}
                 4 + t + 4t^2 & \mbox{if $\set*{\hT_i,\hT_j}$ is not centric}\\
                 s + 3 + t + 2t^2 & \mbox{otherwise} 
            \end{array}
            \right.
        \]
    \end{claim}

    \begin{proof}
        Let $S_i, S_j$ be the type-$k$ stars respectively in $\hT_i$ and in $\hT_j$ that is common to $\set{T_i,T_j}$. Denote by $S' = \hT_i \setminus S_i$ and $S'' = \hT_j \setminus S_j$ the two other stars forming $\hT_i$ and $\hT_j$. As $\hT_i$ and $S_i \cup S'$ differ by an edge, we have $V(\hT_i) = V(S_i) \cup V(S')$, and similarly, $V(\hT_j) = V(S_j) \cup V(S'')$.
        
        Note that by definition of the graph intersection operator, $V(G\cap H) = V(G) \cap V(H)$, for all graphs $G,H$. By distributive law of intersection over union set, we have:
        \begin{eqnarray*}
            &&V(\hT_i \cap \hT_j) \\
            &=& V(\hT_i) \cap V(\hT_j)\\
                &=& \pare{V(S_i) \cup V(S')} \cap \pare{V(S_j) \cup V(S'')}\\
                &=& \pare{V(S_i)\cap V(S_j)} \cup \pare{V(S')\cap V(S_j)} \cup \pare{V(S_i)\cap V(S'')} \cup \pare{V(S')\cap V(S'')}\\
                &=& V(S_i\cap S_j) \cup V(S'\cap S_j) \cup V(S_i\cap S'') \cup V(S'\cap S'')
        \end{eqnarray*}
        In particular,
        \begin{equation}\label{eq:SiSj}
            \cardV{\hT_i \cap \hT_j} ~\le~ \cardV{S_i\cap S_j} + \cardV{S'\cap S_j} + \cardV{S_i\cap S''} + \cardV{S'\cap S''}
        \end{equation}
        Because $i\neq j$, among the four stars $S_i,S_j,S',S''$ contained in $\hT_i \cup \hT_j$, only $S_i$ and $S_j$ have the same type, namely $k$. Let $k',k''$ be the types of $S'$ and $S''$ respectively. So, the three types $k,k',k'' \in\set{1,2,3}$ are pairwise different.
        
        \cref{claim:inter_S} asserts that the of intersection two type-$p$ and type-$p'$ stars contains at most $1 + t^{p+3-p'}$ vertices, assuming $p<p'$. This is always at most $1+t^2$ (as $p'-p \ge 1$), and at most $1+t$, if $p=1$ and $p'=3$.
        
        One can apply on the last three summands of \cref{eq:SiSj} the second upper bound of \cref{claim:inter_S}, since, for each of these terms, the types of the two stars differ. Whatever the relative order of $k,k',k''$, one value is $1+t$, and the other two are $1+t^2$. Therefore, \cref{eq:SiSj} can be rewritten as:
        \begin{eqnarray*}
            \cardV{\hT_i \cap \hT_j} &\le& \cardV{S_i\cap S_j} + (1+t) + 2(1+t^2) \\
            &\le& \cardV{S_i\cap S_j} + 3 + t + 2t^2
        \end{eqnarray*}

        Let $c_i,c_j$ be the centers of $S_i,S_j$ respectively. By definition of $S_i,S_j$, $\set*{\hT_i,\hT_j}$ is centric if and only if $c_i = c_j$. So, if $\set*{\hT_i,\hT_j}$ is not centric, then the first upper bound of \cref{claim:inter_S} applies, and $\cardV*{S_i \cap S_j} \le 1 + 2t^{3-k} \le 1 + 2t^2$ (recall that $S_i,S_j$ are both type-$k$ stars). Otherwise, one can simply bound $\cardV*{S_i\cap S_j} \le \cardV*{S_i} = s$.
        
        Overall, we have proved that $\cardV*{\hT_i,\hT_j} \le (1+2t^2) + 3 + t + 2t^2 = 4 + t + 4t^2$ if $\set*{\hT_i,\hT_j}$ is not centric, and $\cardV*{\hT_i,\hT_j} \le s + 3 + t + 2t^2$ otherwise, as required.
    \end{proof}

    We are now ready to show that $\sum_{i\neq j} \cardV*{\hT_i \cap \hT_j} < 3s$.
    
    By \cref{claim:centric}, at least one pair of trees out of the three possible ones is not centric, and the first upper bound of \cref{claim:inter_T} applies for that pair. For the other two pairs, the second upper bound of \cref{claim:inter_T} applies. It follows that:
    \begin{eqnarray*}
        \sum_{i\neq j} \cardV{\hT_i \cap \hT_j} &\le& \pare{4+t+4t^2} + 2\pare{s + 3 + t + 2t^2} \\
        &\le& 2s + 10 + 3t + 8t^2 
    \end{eqnarray*}
    We check that $10 + 3t + 8t^2 < 1 + t^3 = s$ as soon as $t \ge 9$. And thus, $\sum_{i\neq j} \cardV*{\hT_i \cap \hT_j} < 3s$, for every $t \ge 9$.
    % \Cyril{J'ai utilisé Maple et tracé les deux courbes qui s'intersectent la dernière fois vers t=8.5.}
    
    This implies $\cardV*{\sT} > 3s$, and thus $\cardV*{\sT} > \cardV*{\sU}$ as required. And this completes the proof of \cref{th:3trees}.
\end{proof}

The smallest example provided by \cref{th:3trees}, for $t = 9$, consists of three trees of $2 \cdot 9^3 + 2 = 1460$ vertices. We suspect that this vertex amount can be lowered significantly, with perhaps more complex arguments. %it not worth the cost in proof length and complexity.

% \Cyril{Je pense qu'on pourrait essayer d'optimiser un peu les valeurs, style reprendre les $d_i,p_i$ et voir les contraintes. Puis, rejouer la preuve avec les solutions trouvées. Cela se joue à 5 endroits. Deux fois dans \cref{claim:inter_S} (les deux formules), deux fois dans \cref{claim:inter_T} et une fois dans l'équation finale à la fin qui sera à optimiser. Cela permettrait, si les arbres ont autour de 10/15 sommets par étoiles de donner un exemple contrêt, et de poser: est-ce le plus petit contre-exemple ?}

% \Cyril{Bon, en fait, je crois qu'on peut prendre simplement pour une étoile de type $i=1,2,3$ un simple $K_{1,s-1-i}$ où une arête a été remplacée par un chemin de longueur $i$. Il s'agit de caterpillars (et donc de pathwidth-1). D'après mes calculs, tout marche dès que $s \ge 22$, soit des caterpillars avec $n=44$ sommets (et il suffit d'avoir $n$ pair). Il faut simplement mettre à jour les \cref{claim:inter_S} et \cref{claim:inter_T}.}

% \Edgar{Je suis pas sûr que ça soit nécéssaire de passer de 164 à 22, enfin je vois pas ce que ça apporte}

% \Francois{On ne peut pas dire que ça serait intéressant de trouver plus petit et ne pas le faire alors qu'on peut.}

%------------------------------------------
\subsection{Consequences for greedy strategies}
\label{subsec:greedy}

\Status{Cyril: J'ai fait une passe jusqu'au status suivant. Je me demande si on ne devrait pas mettre des propositions, histoire de ne pas rester sur uniquement de la discussion. D'un autre coté, ces résultats (négatifs) sont de bonnes remarques, sans plus. Méritent-il mieux ? Bon, finallement, j'ai terminé par une proposition.}

As we will see, \cref{th:2trees,th:3trees} imply that there exist graph families where no greedy strategies can build a minimum isometric-universal graph. Informally, a greedy strategy consists in starting with some initial graph of a given family $\sF$, and then, at each step, picking a new graph of $\sF$ and computing a minimum isometric-universal graph for the pair composed of the current graph and this new graph, and updates the current graph with this new isometric-universal graph.

\def\algoA{\textsf{A}}

More formally, let $\algoA$ be any algorithm (considered as a black-box or oracle) computing an isometric-universal graph for a pair of graphs. A greedy strategy, parametrized by $\algoA$, selects the graphs of a given family $\sF$ in some order, say $G_1,\dots,G_{|\sF|}$, and builds the following sequence of graphs $\sU_1,\dots,\sU_{|\sF|}$:
\[
    \sU_i = \left\{%
    \begin{array}{lll}
        G_1 & \mbox{if $i=1$}\\
        \mbox{\algoA}(\sU_{i-1}, G_i) & \mbox{if $i>1$}\\
    \end{array}
    \right.
\]
By construction, $\sU_{|\sF|}$ is an isometric-universal graph for $\sF$, since, by induction, $\sU_i$ is an isometric-universal graph for $\range*{G_1}{G_i}$.

Unfortunately, no greedy strategies can determine the minimum isometric-universal graph for all families of trees, if the routine $\algoA$ constructs a minimum and minimal isometric-universal graph for any two graphs. (Recall that, for two general graphs, this is NP-hard.) This is due to the fact that, at every step $i$, $\algoA(\sU_{i-1}, G_i)$ will return a tree, since, by \cref{th:2trees}, every minimum and minimal isometric-universal graph of two trees has to be a tree. In particular, the final graph $\sU_{|\sF|}$ will be a tree.

Now, by \cref{th:3trees}, there is a family of trees whose minimum isometric-universal graph cannot be a tree. So, it implies that whatever the order of the trees chosen for the greedy strategy, the result cannot be a minimum isometric-universal graph.

Since by \cref{th:2trees}, $\algoA(\sU_{i-1}, G_i)$ must be a tree, it is legitimate to ask whether, for tree families, greedy strategies can help in computing a minimum isometric-universal \emph{tree}, rather than a minimum isometric-universal graph. Unfortunately, this will fail as well for the following reason.

Consider the family $\sF = \set{T_1,T_2,T_3}$, where $T_i$ is composed of two copies of a $K_{1,r}$, whose centers are connected by a path of length $(i+2)s$, where $r$ and $s = o(r)$ are large enough. See \cref{fig:badtrees} for an example.

\begin{myfigure}
    %\ARBRE{px = position en x)}{length = 3,4,5}
    \newcommand{\ARBRE}[2]{%
    \tikzset{every node/.style={vertex}}
    \def\L{0.85} % longeur des arêtes du chemin

    % les sommets de T_i
    \foreach \i in {0,...,#2} \node (p\i) at ({\L*\i},1) {};
    \foreach \i in {0,...,4} {
        \node (x\i) at ($(0,0)+(-0.75,0)!\i/5!(1,0)$) {};
        \node (y\i) at ($({\L*#2},0)+(-0.75,0)!\i/5!(1,0)$) {};
      }
    
    % le chemin
    \foreach \i in {1,...,#2} \draw (p\the\numexpr\i-1) -- (p\i);

    % les étoiles
    \graph{ (p0)  -- { (x0), (x1), (x2), (x3), (x4) } };
    \graph{ (p#2) -- { (y0), (y1), (y2), (y3), (y4) } };

    % nom de l'arbre
    \tikzset{every node/.style={draw=none}}
    \node at (-0.6,1.5) {$T_\the\numexpr#2-2$};

    % accolade sous étoiles
    \tikzset{every node/.style={draw=none,midway}}
    \tikzset{every path/.style={accolade}}
    \draw (x4.east) -- (x0.west) node[yshift=-1.8em] {$r$}; 
    \draw (y4.east) -- (y0.west) node[yshift=-1.8em] {$r$}; 

    % longueur du chemin  
    \draw [draw=none] (p0) -- (p#2) node[yshift=1em] {$#2s$}; 
    }
    \begin{tikzpicture}[scale=0.8]
        \ARBRE{0}{3}
    \end{tikzpicture}
    \hfill
    \begin{tikzpicture}[scale=0.8]
        \ARBRE{0}{4}
    \end{tikzpicture}
    \hfill
    \begin{tikzpicture}[scale=0.8]
        \ARBRE{0}{5}
    \end{tikzpicture}
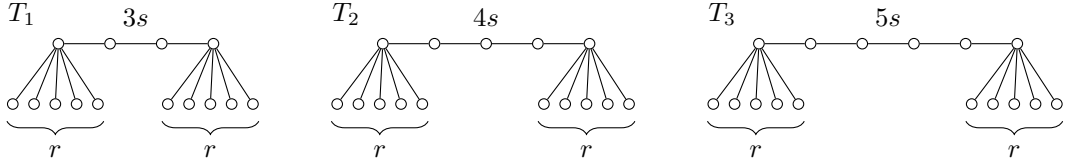
    \caption{A family $\set{T_1,T_2,T_3}$ of trees, here for $r=5$ and $s=1$, whose any greedy strategy designed for trees will fail.}
    \label{fig:badtrees}
\end{myfigure}

It is not difficult to see that the minimum isometric-universal graph of $\set{T_i,T_j}$, i.e., $\algoA(T_i,T_j)$, must contain three stars, i.e., three copies of $K_{1,r}$ intersecting only within a constant number of vertices. Indeed, assuming\footnote{As $\algoA(T_i,T_j) = \algoA(T_j,T_i)$, the order of the two first trees does not matter.} $i<j$, $\algoA(T_i,T_j)$ must be composed of $T_j$ where the center of a star of $T_i$ belongs to the path of $T_j$, the other star being in common. In this universal tree, the distances between these three centers are $(2+i)s$, $(2+j)s$, and $(j-i)s$. For instance, for $i=1$ and $j=3$, we have the distance set $\set{3s,5s,2s}$. The two other distance sets being $\set*{3s,4s,s}$ (for $i,j=1,2$) and $\set*{4s,5s,s}$ (for $i,j=2,3$).

Now, by adding the third tree $T_k$ to $\algoA(T_i,T_j)$, by computing $\algoA(\algoA(T_i,T_j),T_k)$, one obtain a universal tree that must contain four stars. Indeed, the only way to get three stars would be to share both stars of $T_k$ with two out of the three stars already in $\algoA(T_i,T_j)$. Unfortunately, the distance $(2+k)s$ between the two star centers of $T_k$ does not belong to the distance set $\set*{(2+i)s, (2+j)s, (j-i)s}$ of $\algoA(T_i,T_j)$. 

The construction as depicted in \cref{fig:goodtree} shows that $\set*{T_1,T_2,T_3}$ has an isometric-universal tree with only three stars, and thus with $3r + O(s)$ vertices, whereas the greedy strategy $\algoA(\algoA(T_i,T_j),T_k)$ gives $4r + O(s)$ vertices, which is bigger for $s = o(r)$ for instance.

\begin{myfigure}
    \begin{tikzpicture}[scale=0.8]
    \tikzset{every node/.style={vertex}}

    % les sommets
    \foreach \i in {0,...,5} \node (p\i) at (\i,1) {};
    \node (p6) at (2,0) {};
    \foreach \i in {0,...,4} {
        \node (x\i) at ($(0,0)+(-0.75,0)!\i/5!(1,0)$) {};
        \node (y\i) at ($(5,0)+(-0.75,0)!\i/5!(1,0)$) {};
        \node (z\i) at ($(2,-1)+(-0.75,0)!\i/5!(1,0)$) {};
      }
    
    % le chemin
    \foreach \i in {1,...,5} \draw (p\the\numexpr\i-1) -- (p\i);

    % les étoiles
    \graph{ (p0) -- { (x0), (x1), (x2), (x3), (x4) } };
    \graph{ (p5) -- { (y0), (y1), (y2), (y3), (y4) } };
    \graph{ (p2) -- { (p6) -- { (z0), (z1), (z2), (z3), (z4) } }};
    
    % longueur des chemins  
    \tikzset{every node/.style={draw=none,inner sep=3pt,midway}}
    \draw [draw=none] (p0) -- (p5) node[yshift=1em] {$5s$}; 
    \draw [draw=none] (p2) -- (p6) node[xshift=1em] {$4s$}; 
    \draw [draw=none] (p2) -- (p6) node[xshift=-1em] {$3s$}; 

    \end{tikzpicture}
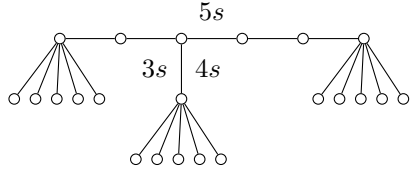
    \caption{An isometric-universal tree for $\set{T_1,T_2,T_3}$, depicted in \cref{fig:badtrees}, with only $3r + O(s)$ vertices.}
    \label{fig:goodtree}
\end{myfigure}

To summarize the above discussion, we have shown:

\begin{proposition}
    Let $\algoA$ be any algorithm computing a minimum and minimal isometric-universal graph for a pair of graphs. Then, no greedy strategies based on $\algoA$ can compute a minimum isometric-universal graph (or tree) for all tree families.
\end{proposition}

We note that we have analyzed greedy strategies for its \emph{best-case} running (and not its worst-case as usual), that is for best ordering of the graphs in the family. So, these greedy strategies fail in a strong sense (i.e. whatever the ordering is).

%\Cyril{Question ouverte: déterminer les facteurs d'approximation pour ces divers algorithmes glouton. Si $\algoA$ a un facteur d'approx $\alpha$, combien fait l'algo greedy?}

\Status{Cyril: J'ai fait une passe jusqu'à la fin de la section et c'est bon pour moi.}

\Status{Edgar : J'ai fait une passe, ça me semble bien. Le seul point qui me semble plus complexe à comprendre c'est la partie où on justifie la borne inf d'un $k$-iso universal tree.}

%------------------------------------------
\subsection{Extending \cref{th:2trees} to $k$-isometric-universal trees}

% Cyril: ne pas mettre cref{}, sinon on a un warning => Balec avec LIPIcs

Another natural extension of \cref{th:2trees} is to ask whether it holds for $k$-isometric-universal graphs, for $k$ large enough and depending of $n$. 
%\Cyril{Je ne pense pas qu'il soit nécessaire de définir un nouveau problème en question. On le fait seulement si (1) on doit se reservir du problème, et (2) si c'est un problème de décision, style pour dire que le problème \textsc{X} est NP-Complet. Là, on souhaite dire un truc différent: Dans quelle mesure peut-on étendre le \cref{th:2trees} ? Bah seulement si $k > n/3$. Basta. On s'en fout de comment nomer ce problème. Dit autrement, de manière très pragmatique, on ne peut pas exprimer la \cref{th:k0} avec le problème ci-dessus. En plus, ce problème me paraît un poils artificiel car on est dans la mouize pour dire que c'est polynomial si jamais le GUI ne fait pas un arbre ... Donc on est coincer avec l'extension du \cref{th:2trees}. C'est vrai que pour $k=1$, c'est NP-Complet. On peut alors se poser la question de tout $k>1$. Certes pour $k=n-1$ c'est polynomial, mais qui se pose le problème de la $k(n)$-coloration ? C'est bien pareil: pour $k=3$ c'est NP-Complet et pour $k=n$, c'est polynomial. }
%
% The decision problem now becomes:
%
%\myproblem%
%{$k(n)$-isometric-universal Graph of Two Trees}%
%{($k(n)$-\IUTT)}%
%{Two trees $T_1,T_2$ and $t \in\mathbb{N}$. Let $n = \max(|V(T_1)|, |V(T_2)|)$.}%
%{Is there a $k(n)$-isometric-universal graph for $\set{T_1,T_2}$ with $t$ vertices?}
More formally, we may ask what is the smallest function $k^*(n)$ defined as:

\begin{myquote}
For every $k \ge k^*(n)$, every minimum and minimal $k$-isometric-universal graph for a pair of trees each with at most $n$ vertices is a tree. 
\end{myquote}

Clearly, every graph $\sU$ containing an $n$-vertex tree $T$ as an $(n-1)$-isometric subgraph, also contains $T$ as an isometric subgraph. Thus, by \cref{th:2trees}, $k^*(n) \le n-1$ for all $n\ge 1$. On the other hand, as we will see hereafter in \cref{th:k0}, there are pairs of $n$-vertex trees for which every minimum $1$-isometric-universal graph is not a tree whenever $n\ge 11$. Thus, $k^*(n) > 1$ for all $n\ge 11$.

We establish the following lower bound for $k^*(n)$:

\begin{theorem}\label{th:k0}
  For every $n\ge 8$, $k^*(n) > \floor{(n-8)/3}$. In other words, there is a pair of trees, each with at most $n$ vertices, for which every minimum $\floor{(n-8)/3}$-isometric-universal graph is not a tree.
\end{theorem}

\begin{proof}
%\Cyril{J'ai l'impression que c'est aussi vrai pour $k=0$. Et donc la \cref{th:k0} pourrait être vraie dès $n\ge 8$ au lieu de $n\ge 11$.}
%\Francois{J'ai peut-être mal saisi les définitions, mais il me semble que les distances sont décalées de 1. Pour $k$-isométrique, il suffit d'avoir distance $k+1$ pour que ça ne soit pas comparable, et pas $k+2$ non ? Calculs numériques à vérifier en fonction.} \Cyril{Je pense que c'est correct -- voir commentaire ci-après dans la preuve. Mais sans doute qu'il faudrait mieux expliquer.}
Let $k \in \range{0}{n-7}$, and define $m$ as the largest integer such that $2m+k+3 \le n$. In other words, $m = \floor{(n-k-3)/2} \ge 0$. For each $i\in\set{0,1}$, we define $T_i(k)$ as the tree obtained from of a path with $k+3-i$ vertices where $m$ leaves are attached to both extremities. Note that the two extremities of the path, where leaves are attached, are thus at distance $k+2-i$ in $T_i(k)$. Then, $T_i(k)$ has a total of $2m+k+3-i \le 2m+k+3 \le n$ vertices by the choice of $m$. See \cref{fig:Ti(k)} for an illustration. Note that $T_1(k)$ can be obtained from $T_0(k)$ by contracting one edge of its path.

%\Edgar{Pour que le chemin central soit de taille 4 et 5 il faut que $k=2$. Sinon (si $k\ge 3$) on peut différencier si la distance entre les extrémités est 3 ou plus (chemin de taille 4 => distance 3 entre les extrémités)} \Cyril{Ok. S'il faut prendre $k=2$, pourquoi donc alors $\sT(k)$ serait-il un arbre optimal pour $k=2$? On veut arriver à dire que $|\sC(k)| < B$ où $B$ est la taille de l'abre optimal. Donc on aimerait bien dire que $B \le |\sT(k)|$ voir que $B = |\sT(k)|$. Il nous faut une borne inf sur $|\sT(k)|$ ... à prouver donc.} \Edgar{Effectivement il faut le montrer, mais je pense qu'on peut le faire simplement en montrant que les deux centres des étoiles ne peuvent pas être partagées car distance différentes ce qui impliquerait deux chemins entre leur centre donc cyelc ce qui fait minimum $3\times($ Taile d'une étoile $)$ + la longueur du plus long chemin et c'est gagné.}

  \begin{myfigure}
    \begin{tikzpicture}[scale=0.82]
    % \draw[help lines] (0,0) grid (15,2);
    \tikzset{every node/.style={vertex}}

    % les sommets T_0
    \foreach \i in {1,...,5} \node (p\i) at (\i,1) {};
    \foreach \i in {0,...,4} {
        \node (x\i) at ($(1,0)+(-1,0)!\i/4!(1,0)$) {};
        \node (y\i) at ($(5,0)+(-1,0)!\i/4!(1,0)$) {};
      }

    % les sommets T_1
    \foreach \i in {1,...,4} \node (q\i) at ($(9,0)+(\i,1)$) {};
    \foreach \i in {0,...,4} {
        \node (u\i) at ($(10,0)+(-1,0)!\i/4!(1,0)$) {};
        \node (v\i) at ($(13,0)+(-1,0)!\i/4!(1,0)$) {};
      }

    % les chemins
    \draw (p1) -- (p2) -- (p3) -- (p4) -- (p5);
    \draw (q1) -- (q2) -- (q3) -- (q4);

    % les étoiles
    \graph{ (p1) -- { (x0), (x1), (x2), (x3), (x4) } };
    \graph{ (p5) -- { (y0), (y1), (y2), (y3), (y4) } };
    \graph{ (q1) -- { (u0), (u1), (u2), (u3), (u4) } };
    \graph{ (q4) -- { (v0), (v1), (v2), (v3), (v4) } };

    % accolades: m sur etoiles et k + 2 ou 3 sur les chemins
    \tikzset{every node/.style={draw=none,midway}}
    \tikzset{every path/.style={accolade}}
    \draw (x4.east) -- (x0.west) node[yshift=-1.8em] {$m$}; 
    \draw (y4.east) -- (y0.west) node[yshift=-1.8em] {$m$}; 
    \draw (u4.east) -- (u0.west) node[yshift=-1.8em] {$m$}; 
    \draw (v4.east) -- (v0.west) node[yshift=-1.8em] {$m$}; 
    \draw (p1.west) -- (p5.east) node[yshift=1.8em] {$k+3$}; 
    \draw (q1.west) -- (q4.east) node[yshift=1.8em] {$k+2$}; 

      % les noms
    \tikzset{every node/.style={draw=none}}
    \tikzset{every path/.style={draw=none}}
    \node at (0,1) [above=4pt] {$T_0(2)$};
    \node at (9,1) [above=4pt] {$T_1(2)$};

    \end{tikzpicture}
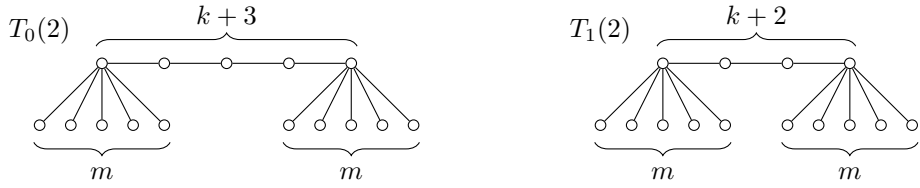
    \caption{The trees $T_0(k)$ and $T_1(k)$ for $k=2$ and $n=15$. The numbers indicated concern number of vertices, and not edges. We have $m = 5$.}
    \label{fig:Ti(k)}
  \end{myfigure}

    For convenience, let $\sF(k) = \set{T_0(k),T_1(k)}$. Let $\sT(k)$ and $\sC(k)$ be two $k$-isometric-universal graphs for the family $\sF(k)$ as follows (see also \cref{fig:TC(k)}). The graph $\sT(k)$ is obtained from $T_1(k)$ where are attached $m$ leaves to one of its leaves. It has $|V(T_1(k))| + m$ vertices. The graph $\sC(k)$ is obtained from $T_1(k)$ where a path of length $k+2$ is attached between the two centers of the stars of $T_1(k)$. It has $|V(T_1(k))| + k+1$ vertices. 

  \begin{myfigure}
    \begin{tikzpicture}[scale=0.78]
    % \draw[help lines] (0,0) grid (15,2);
    \tikzset{every node/.style={vertex}}

    % les sommets
    \foreach \i in {1,...,4} \node (p\i) at (\i,1) {};
    \foreach \i in {1,...,4} \node (q\i) at ($(9,0)+(\i,1)$) {};
    \foreach \i in {1,...,3} \node (r\i) at ($(9.5,0.5)+(\i,1)$) {};
    \foreach \i in {0,...,4} {
        \node (x\i) at ($(1,0)+(-1,0)!\i/4!(1,0)$) {};
        \node (y\i) at ($(4,0)+(-1,0)!\i/4!(1,0)$) {};
        }
    \foreach \i in {0,...,4} {
        \node (z\i) at ($(5.5,1)+(-1,0)!\i/4!(1,0)$) {};
      }

    % les chemins
    \draw (p1) -- (p2) -- (p3) -- (p4);

    % les étoiles
    \graph{ (p1) -- { (x0), (x1), (x2), (x3), (x4) } };
    \graph{ (p4) -- { (y0), (y1), (y2), (y3), (y4) } };
    \graph{ (y4) -- { (z0), (z1), (z2), (z3), (z4) } };

    \foreach \i in {1,...,4} \node (q\i) at ($(9,0)+(\i,1)$) {};
    \foreach \i in {0,...,4} {
        \node (u\i) at ($(10,0)+(-1,0)!\i/4!(1,0)$) {};
        \node (v\i) at ($(13,0)+(-1,0)!\i/4!(1,0)$) {};
        }

    % les chemins
    \draw (q1) -- (q2) -- (q3) -- (q4);
    \draw (q1) -- (r1) -- (r2) -- (r3) -- (q4);
    
    % les étoiles
    \graph{ (q1) -- { (u0), (u1), (u2), (u3), (u4) } };
    \graph{ (q4) -- { (v0), (v1), (v2), (v3), (v4) } };

    % accolades: m sur etoiles et k + 2 ou 3 sur les chemins
    \tikzset{every node/.style={draw=none,midway}}
    \tikzset{every path/.style={accolade}}

    \draw (x4.east) -- (x0.west) node[yshift=-1.8em]{$m$}; 
    \draw (y4.east) -- (y0.west) node[yshift=-1.8em]{$m$}; 
    \draw (u4.east) -- (u0.west) node[yshift=-1.8em]{$m$}; 
    \draw (v4.east) -- (v0.west) node[yshift=-1.8em]{$m$}; 
    \draw (z0.west) -- (z4.east) node[yshift=1.8em]{$m$}; 
    \draw (p1.west) -- (p4.east) node[yshift=1.8em]{$k+2$}; 

    % les noms
    \tikzset{every node/.style={draw=none}}
    \tikzset{every path/.style={draw=none}}
    \node at (0,1) [above=4pt] {$\sT(2)$};
    \node at (9,1) [above=4pt] {$\sC(2)$};
    \node at (11.5,0.6) {$k+2$};
    \node at (11.5,1.8) {$k+3$};

    \end{tikzpicture}
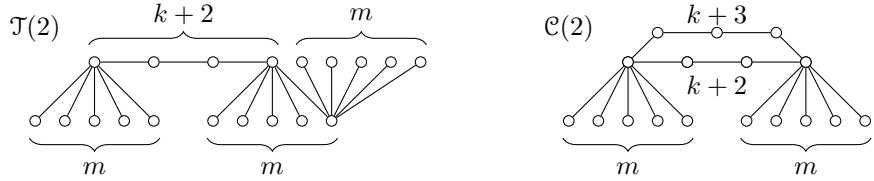
    \caption{The two graphs $\sT(k),\sC(k)$ are $k$-isometric-universal for $\set{T_0(k),T_1(k)}$, and $\sT(k)$ is a minimum $k$-isometric-universal tree, here for $k = 2$ and $n = 15$. The numbers indicated concern number of vertices, and include the extremities of the path in $\sC(k)$.
    }
    \label{fig:TC(k)}
  \end{myfigure}

    It is not difficult to see that $\sT(k)$ is an isometric-universal tree for $\sF(k)$, and thus is $k$-isometric-universal. We check that $\sC(k)$ is $k$-isometric-universal graph for $\sF(k)$ too, but is not $(k+1)$-isometric-universal. Indeed, for every embedding of $T_0(k)$ in $\sC(k)$, the two centers of the stars of $T_0(k)$ are at distance $k+1$ in $\sC(k)$, whereas they are at distance $k+2$ in $T_0(k)$. (Recall that $H$ is a $k'$-isometric subgraph of $G$ if $\dist_G(u,v) \le k'$ implies $\dist_H(u,v) = \dist_G(u,v)$ for all $u,v \in V(H)$.)
    \Cyril{En effet, c'est bien correct, car la définition dit qu'un sous-graphe $H$ de $G$ est $k'$-isométrique si, pour tout $u,v \in H$, $\dist_G(u,v) \le k'$ implique $\dist_H(u,v) = \dist_G(u,v)$. En particulier (et c'est une remarque de la définition au début de l'article), si $\dist_H(u,v) = k'+1$, on doit forcément avoir $\dist_G(u,v) = k'+1$. Ici, $G = \sC(k)$, $H = T_0(k)$ et $k' = k+1$. Donc, si $\sC(k)$ était $(k+1)$-isometric-universal, alors la distance dans $T_0(k)$ entre les centres $c_1,c_2$, qui vaut $\dist_{T_0(k)}(c_1,c_2) = k+2$, devrait impliquer que $\dist_{\sC(k)}(c_1,c_2) = k + 2$, ce qui n'est pas le cas puisque $\dist_{\sC(k)}(c_1,c_2) = k+1$.}
    \Francois{Ok, mais je trouve que parler de distances et de nombres de sommets prête à confusion, surtout dans les figures. J'ai rajouté quelques précisions dans la définition et dans les légendes des figures dans ce sens.}
    
    Let us show that:
    \begin{claim}\label{claim:Tk_min}
        $\sT(k)$ is a minimum $k$-isometric-universal tree for $\sF(k)$.
    \end{claim}

    \begin{proof}
    Each tree $T_i(k)$ contains two stars, that are copies of a $K_{1,m+1}$ (with one leaf being part of the path connecting the centers). Moreover, the centers of the stars are connected by paths of different lengths in $T_0(k)$ and in $T_1(k)$, namely $k+1$ and $k+2$. 

    Let $\sU$ be any minimum $k$-isometric-universal tree for $\sF(k)$. $\sU$ contains $T_0(k)$ as subgraph. If the two centers of the stars of $T_1(k)$ are embedded at the two centers of $T_0(k)$ is $\sU$, then $\sU$ would contain two different paths connecting these two common centers, respectively of length $k+1$ and $k+2$, implying a cycle of length at least three in $\sU$. Therefore, the centers of the stars in $T_0(k),T_1(k)$ must be embedded in at least three different vertices of $\sU$. Now, it is easy to see that two stars not sharing the same center in $\sU$ can share at most two vertices (along the same edge), creating a cycle otherwise. And, a star can share with a path at most three vertices (along two incident edges).

    Altogether, $\sU$ must contain $T_0(k)$ and a copy of $K_{1,m+1}$ sharing at most three vertices with $T_0(k)$. It follows that the number of vertices of $\sU$ is at least
    \[
        |V(T_0(k))| + |V(K_{1,m+1})| - 3 = |V(T_0(k))| + m-1 = |V(T_1(k))| + m ~.
    \]
    This is exactly the number of vertices of $\sT(k)$, which is, we have seen, a $k$-isometric-universal tree for $\sF(k)$. Therefore, $\sT(k)$ is minimum.
    \end{proof}
    
    Consider any $k$ such that $|V(\sC(k))| < |V(\sT(k))|$. Since $\sC(k)$ is a $k$-isometric-universal graph for $\sF(k)$, \cref{claim:Tk_min} implies that, for such $k$, every minimum $k$-isometric-universal graph for $\sF(k)$ is not a tree. Now, the largest integer $k$ such that $|V(\sC(k))| \le |V(\sT(k))| - 1$, i.e., such that $|V(T_1(k))| + k+1 \le |V(T_1(k))| + m-1$, is $k = m-2$. This is $k = \floor{(n-k-3)/2} - 2 = \floor{(n-k-7)/2}$, which is non-negative since $k \le n-7$. Depending whether $n-k-7$ is odd or even, we have $k \in \set{ \floor{(n-8)/3}, \floor{(n-7)/3} }$. 

    Therefore, for $k = \floor{(n-8)/3}$, which is non-negative for $n \ge 8$, the minimum $k$-isometric-universal graph for $\sF(k)$ is not a tree. It follows that $k^*(n) > k$.
\end{proof}

Despite \cref{th:k0} showing that $k^*(n) \ge n/3 - O(1)$, we suspect that $k^*(n) = n/2 + o(n)$, and we leave open to determine the right asymptotic value for $k^*(n)$. On one side we believe that the above proof can be enhanced to reach $k^*(n) \ge n/2 - o(n)$. On the other hand, in the $k^*(n)$-isometric-universal graph, it seems far-fetched to have two disjoint paths of length $> n/2$.

\subsection{NP-completeness}
\label{sec:NPC}

% Cyril: voici la def.
%
% A problem is said to be NP-hard if everything in NP can be transformed in polynomial time into it even though it may not be in NP. A problem is NP-complete if it is both in NP and NP-hard. The NP-complete problems represent the hardest problems in NP. NP-hard problems are not necessarily decision problems, whereas NP-complete are exclusively decision problems. Not all NP-hard problems are NP-complete.

In this section we show that \cref{th:2forests} cannot be extended to more than two forests (unless P = NP). In several decision problems, setting the input from~$2$ to~$3$ costs to pass from P to NP-C (for instance $2$-SAT has a polynomial solution and $3$-SAT is NP-complete, same for graph $2/3$-coloring). We present here the same pattern with the problem of isometric-universal graph from~$2$ to~$3$ forests. As a byproduct of our proof, the NP-completeness holds if we restrict the isometric-universal graph to be a forest. This restriction has also been claimed by Rautenbach and Werner~\cite[Theorem~3]{RW24}.

\begin{theorem}\label{th:np-hard}
    Determining whether there exists an isometric-universal graph (or forest) with~$t$ vertices for three forests is NP-complete, even if the three forests have pathwidth at most two.
\end{theorem}

\begin{proof}
    It is not difficult to see that this problem is in NP.  We use a reduction from the \MATCH problem that is NP-complete even if each element is contained in at most three sets (see \cite[Problem SP1 in Appendix A.3.1]{GJ79}). This restriction can be seen as determining a perfect matching in a 3-partite sub-cubic hypergraph.
    
    For the \MATCH problem, we are given three disjoint sets $X,Y,Z$ of size $n$, and a set of triples $T \subseteq X\times Y\times Z$. We ask whether there exists or not a subset of $n$ triples of $T$ that are pairwise independent. In other words, a solution for the problem, called \emph{matching}, is a subset $M \subseteq T$ of size $n$ such that no two triples of $M$ agree in any coordinate.

    In order to encode the \MATCH problem into our problem, consider a instance $(X,Y,Z,T)$ where each element of $X,Y$, or $Z$ is contained in two or three triples\footnote{Clearly, any element has to appear at least once in $T$, and if it appears only once, then we have to take it in any solution, and so reducing the instance.} of $T$. We have $|X| = |Y| = |Z| = n$ and $|T| \in \range{2n}{3n}$, so that the instance is of size $O(n)$.
    
    We construct a family $\sF = \set{F^X,F^Y,F^Z}$ composed of three forests, each with $n$ trees, one tree for each element of $X$, $Y$, and $Z$. More precisely, for each set $W \in \set{X,Y,Z}$, the set of connected components of $F^W$ is denoted by $\set{ F(w) : w\in W}$, where $F(w)$ is the tree associated with the element $w \in W$ of the forest $F^W$. Each tree $F(w)$ will be described precisely hereafter.
    
    For every $w\in\range{1}{2n}$, let $S(w)$ be a subdivided star obtained from a $K_{1,3w+3}$ whose three of its edges have been replaced by a path of length $2n+2-w$. Let $\sS = \range{S(1)}{S(2n)}$ be the collection of all such stars that are pairwise non-isomorphic. See \cref{fig:S2} for examples of such stars.

    \begin{myfigure}
    \newcommand{\STAR}[1]{%
    \pgfmathtruncatemacro\N{2} % n
    \pgfmathtruncatemacro\D{3*#1+3} % degré
    \pgfmathtruncatemacro\L{2*\N+2-#1} % longueur
    \def\sep{0.3}% écart entre deux voisins de c

    % les sommets de l'étoile
    \tikzset{every node/.style={vertex}}
    \pgfmathsetmacro\x{(\sep*(\D-1)/2};
    \node (c) at (\x,-1.3) {};
    
    \foreach \i in {1,...,\D} {
        \pgfmathsetmacro\x{\sep*(\i-1)}
        \node (p\i) at (\x,0) {};
        \draw (c) -- (p\i);
    };

    \foreach \b in {1,2,3} {
        \pgfmathsetmacro\x{\sep*(\b-1)}
        \foreach \i in {1,...,\L} {
            \pgfmathsetmacro\y{\i-1}
            \node (q\b\i) at (\x,\y) {};
        };
        \foreach \i in {2,...,\L} {
            \pgfmathtruncatemacro\j{\i-1}
            \draw (q\b\i) -- (q\b\j);
        };
    };

    % nom de l'étoile
    \tikzset{every node/.style={draw=none}}
    \node[xshift=-1.8em,yshift=-0.6em] at (c) {$S(#1)$};

    % accolade sur l'étoile
    \tikzset{every node/.style={draw=none,midway}}
    \tikzset{every path/.style={accolade}}
    \draw (p4.west) -- (p\D.east) node[yshift=1.8em] {$\the\numexpr\D-3$};
    \draw (q11.south) -- (q1\L.north) node[xshift=-1.8em] {$\L$}; 
    }
    \begin{tikzpicture}[scale=0.8]
        \STAR{1}
    \end{tikzpicture}
    \hfill
    \begin{tikzpicture}[scale=0.8]
        \STAR{2}
    \end{tikzpicture}
    \hfill
    \begin{tikzpicture}[scale=0.8]
        \STAR{3}
    \end{tikzpicture}
    \hfill
    \begin{tikzpicture}[scale=0.8]
        \STAR{4}
    \end{tikzpicture}   
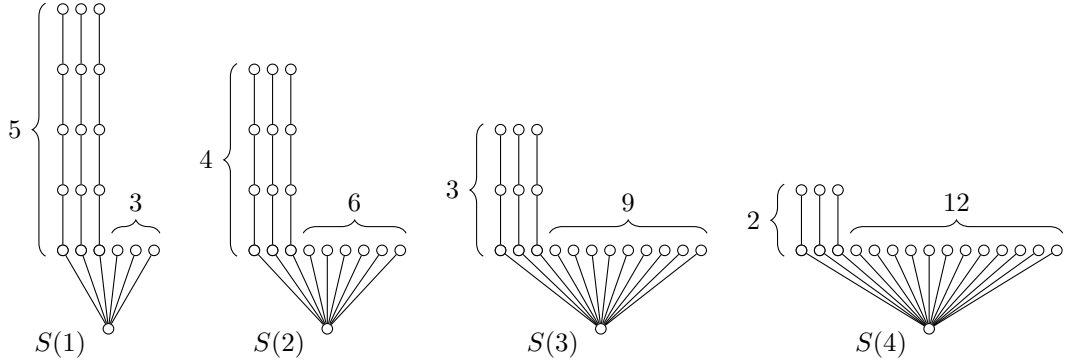
    \caption{The set $\sS = \range{S(1)}{S(2n)}$ of subdivided stars for $n=2$.}
    \label{fig:S2}
    \end{myfigure}
    
    It is easy to check that:

    \begin{claim}\label{claim:S(w)}
        Each $S(w) \in \sS$ is a tree of pathwidth two\footnote{We can use the same argument we used on page~\pageref{page:pw2} to show that $T_i$ was of pathwidth two.} with $3n+1$ vertices, its center has degree $3w+3$, $3w$ branches of length one, and three branches of length $2n+2-w \ge 2$.
    \end{claim}

    \myparagraph{The forests $F^Y$ and $F^Z$.} For convenience, we will assume that $Y = \range{1}{n}$ and $Z = \range{n+1}{2n}$, so that $Y\cup Z = \range{1}{2n}$. Actually, any one-to-one mapping from $Y\cup Z$ to $\range{1}{2n}$ is fine. For each $w\in Y\cup Z$, the tree $F(w)$ is composed of a copy of the star $S(w) \in \sS$ where is attached a \emph{claw}, i.e., a copy of a $K_{1,3}$. More precisely, the center of $S(w)$ is merged with a leaf of the claw. The tree $F(w)$ has $\cardV*{S(w)} + \cardV*{K_{1,3}} - 1 = 3n+4$ vertices, and pathwidth two. See \cref{fig:F(w)}.

    \begin{myfigure}
    \begin{tikzpicture}[scale=0.9]
    \def\sep{0.4}% écart entre deux voisins du centre
    \tikzset{every node/.style={vertex}}

    % les sommets du K_{1,3}
    \pgfmathsetmacro\x{((\sep*7)/2};
    \node (c) at (\x,-1) {};
    \node (x) at (\x,-1.7) {};
    \node (y) at (\x-0.75,-2) {};
    \node (z) at (\x+0.75,-2) {};

    % les voisins du centre
    \foreach \i in {1,...,5,7,8} {
        \pgfmathsetmacro\x{\sep*(\i-1)}
        \node (p\i) at (\x,0) {};
        \draw (c) -- (p\i);
    };
    \draw[dotted] (p5) -- (p7);

    % claw
    \draw[red,line width=1pt] (y) -- (x) -- (z);
    \draw[red,line width=1pt] (c) -- (x);

    % les trois longs chemins
    \foreach \b in {1,2,3} {
        \pgfmathsetmacro\x{\sep*(\b-1)}
        \node (q\b1) at (\x,0) {};
        \node (q\b2) at (\x,1) {};
        \coordinate (q\b3) at (\x,1.75) {};
        \node (q\b4) at (\x,2.5) {};
        \draw (q\b1) -- (q\b2);
        \draw[dotted] (q\b2) -- (q\b3);
        \draw (q\b3) -- (q\b4);
    };

    % accolades
    \tikzset{every node/.style={draw=none,midway}}
    \tikzset{every path/.style={accolade}}

    %\tikzset{every path/.style={decorate,decoration={brace,amplitude=5pt,raise=1.5ex}}}
    
    \draw (p4.west) -- (p8.east) node[yshift=1.8em] {$3w$};
    
    \draw (q11.south) -- (q14.north) node[xshift=-4em] {$2n+2-w$};
    
    \draw (3,-1.05) -- (3,-2) node[xshift=2.5em] {$K_{1,3}$};

    \draw (3,2.6) -- (3,-0.95) node[xshift=2.5em] {$S(w)$}; 

    % nom du graphe
    \tikzset{every node/.style={draw=none}}
    \tikzset{every path/.style={draw}}
    \node[shift={(-2,1)}] at (y) {$F(w)$};
    \end{tikzpicture}
    
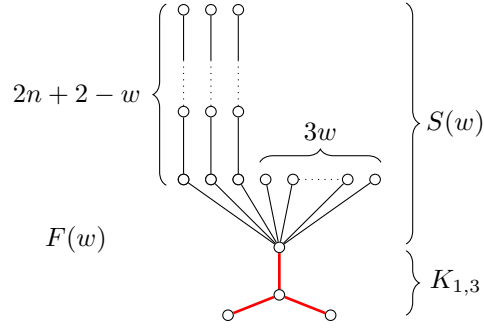
    \caption{The tree $F(w)$ associated with element $w \in Y\cup Z = \range{1}{2n}$ constructed from $S(w) \in \sS$ and a claw $K_{1,3}$ (in red).}
    \label{fig:F(w)}
    \end{myfigure}

    \myparagraph{The forest $F^X$.} Consider any $x\in X$, and define $T(x) = \set{ (x,y,z) \in T }$ to be the set of triples of $T$ containing element $x$. %\Edgar{triple or pair ?} \Cyril{On pourrait faire avec les paires seulement effectivement, mais cela me paraît plus simple avec les triplets.}\Edgar{Je voulais dire que en disant triple on a l'impression qu'il n'y a que des triplets, alors qu'il y a des paires} \Cyril{Et bien c'est toujours des triples: seulement il y en a deux ou trois, c'est-à-dire $|T(x)| \in\set{2,3}$}.
    From the restriction of the input, $|T(x)| \in \set{2,3}$.

    Let us denote by $y_1,\dots,y_t$ and $z_1,\dots,z_t$ the elements of $Y$ and $Z$ that are contained in the triples of $T(x)$, where $t = |T(x)|$. The tree $F(x)$ of the forest $F^X$ is obtained from the copies of the $2t$ stars $S(y_1),\dots,S(y_t),S(z_1),\dots,S(z_t)$ where their centers are connected by a path in which each edge is alternatively replaced by a path of length two or three. Furthermore, along the path, the $2t$ star centers are encountered in the specific order $S(y_1) -\circ- S(z_1) -\circ-\circ- S(z_2) -\circ- S(y_2)$ if $t=2$ (dashes mean edges), and in the order $S(y_1) -\circ- S(z_1) -\circ-\circ- S(z_2) -\circ- S(y_2) -\circ-\circ- S(y_3) -\circ- S(z_3)$ if $t=3$. See \cref{fig:F(x)} with the notation therein.
    
    \begin{myfigure}
    \begin{tikzpicture}[scale=0.72]
    \def\sep{3}% écart entre deux centres
    \def\S{{"","y_1","z_1","z_2","y_2","y_3","z_3"}} % liste des stars
    \def\H{{0,.9,.3,.7,.9,.5,1}} % liste des hauteurs

    \tikzset{every node/.style={vertex}}
    \def\GET#1{\pgfmathparse{#1}\pgfmathresult}

    % dessine les triangles
    \foreach \i in {1,...,6} {
        \pgfmathsetmacro\x{\sep*(\i-1)}
        \pgfmathparse{\H[\i]}
        \pgfmathsetmacro\h{\pgfmathresult}
        \pgfmathsetmacro\z{2.5-\h}
        \coordinate (c\i) at (\x,0);
        \node[draw=none,yshift=-1.1em] at (c\i) {$c(\GET{\S[\i]})$};
        \node[draw=none,yshift=1.7em] at ($(c\i)+(0,2)$) {$S(\GET{\S[\i]})$};
        \draw[rounded corners=4pt] (c\i) -- ++(120:2.5) -- ++(0.5,0) -- ++(-60:\h) -- ($(60:\z)+(c\i)$) -- (c\i);
    };

    \foreach \i in {1,...,5} {
        \pgfmathtruncatemacro\j{\i+1} % j=i+1
        \pgfmathtruncatemacro\m{Mod(\i,2)} % m=i%2
        \ifthenelse{\i = 2}
        {\def\P{$p(\GET{\S[\i]},\GET{\S[\j]})$}}
        {\def\P{$p(\GET{\S[\j]},\GET{\S[\i]})$}}
        \ifthenelse{\m = 1}
        {
            \coordinate (u\i) at ($(c\i) + (0.5*\sep,0)$); % sommet du milieu
            \node[draw=none,yshift=1em] at (u\i) {\P};
            \draw (c\i) -- (u\i) -- (c\j);
        }
        {
            \coordinate (u\i) at ($(c\i) + (0.35*\sep,0)$);
            \coordinate (v\i) at ($(c\i) + (0.65*\sep,0)$);
            \draw (c\i) -- (u\i) -- (v\i) -- (c\j);
            \node at (v\i) {};
        }
        \node at (u\i) {};
        \node at (c\i) {};
    };
    \node at (c6) {};

    \end{tikzpicture}
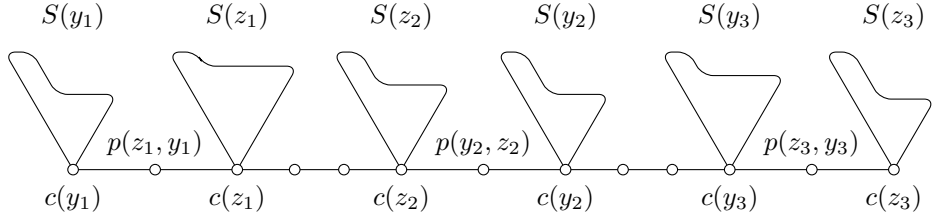
    \caption{The tree $F(x)$ associated with the element $x \in X$. In this example, $T(x) = \set{(x_1,y_1,z_1),(x_2,y_2,z_2),(x_3,y_3,z_3)}$, and $F(x)$ is composed of six distinct stars of $\sS$. Vertex $c(w)$ is the center of the copy of $S(w)$ in $F(x)$, and $p(w,w')$ with $w<w'$ is the unique vertex of the path between the centers $c(w)$ and $c(w)$, if $(x,w,w') \in T(x)$.}
    \label{fig:F(x)}
    \end{myfigure}

    Clearly, $F(x)$ is a tree of pathwidth two. Since each star $S(w)$ has $3n+1$ vertices (cf. \cref{claim:S(w)}), the number of vertices of $F(x)$ is $\cardV*{F(x)} = 2t \cdot (3n+1) + 3t-2 = 6tn + 5t-2 \in \set{12n+8,18n+13}$ depending on $t$. By construction, we have $\cardV*{F^X} = \sum_{x\in X} \cardV{F(x)}$. 
    
    An important property of $F(x)$ is that if $(x,y,z) \in T$, then a copy of the stars $S(y)$ and $S(z)$ appears as an isometric subgraph of $F(x)$. However, neither $F(y)$ nor $F(z)$ can appear in $F(x)$ as an isometric subgraph because of their claws. Another important property is that if $(x,y,z) \notin T$, then the stars $S(y)$ and $S(z)$ cannot appear consecutively in the main path of $F(x)$ (if they appear). This is due to the specific order of the stars along the path. In that case, the distance between their centers is at least 5.
    
    Let $\sU$ be any minimum isometric-universal graph for $\sF = \set{F^X,F^Y,F^Z}$. We note that $\sF$ has polynomial size w.r.t. the instance $(X,Y,Z,T)$, since each of the forests $F^X,F^Y,F^Z$ is composed of $n$ trees with $O(n)$ vertices each. 
    
    \myparagraph{Intuition behind the proof.} Each component of $\sU$ must contain the trees $F(x)$, $F(y)$ and $F(z)$ for some elements $x,y,z$. In a favorable case, this component is composed of $F(x)$ plus only one single vertex. Indeed, if $y,z$ are both contained in $T(x)$, then $S(y)$ and $S(z)$ each appears in $F(x)$. It follows that $F(y)$ and $F(z)$ will almost fit in $F(x)$ but for one vertex of their claw. The only way to pack in $F(x)$ using one extra vertex is to share this vertex with their claw by being consecutive on the main path of $F(x)$. This favorable situation will occur simultaneously on each $F(x)$ if $(X,Y,Z,T)$ has a solution, i.e., if $T$ contains a matching for $X\times Y\times Z$. Unfortunately, if $y$ or $z$ are not in $T(x)$, and still want to pack in $F(x)$, then the resulting component would contain more than one vertex more as in $F(x)$. This situation will necessarily occur if $(X,Y,Z,T)$ has no solution. Overall, the number of vertices in $\sU$ will allow us to determine whether a solution exists or not for $(X,Y,Z,T)$.
        
    More formally, our goal is to prove the following claim:

\begin{claim}\label{claim:goal}
    $\cardV{\sU} \le \cardV{F^X} + n$ if and only if $(X,Y,Z,T)$ has a solution. %\Edgar{to the 3D-Matching problem} \Cyril{Mmmouais. Pour moi, soit on dit que $X\times Y\times Z$ a un matching dans $T$, soit on dit que l'input $(X,Y,Z,T)$ a une solution (synonyme d'instance positive). Bon, on est dans l'énoncé d'un Claim dans une preuve, pas grave si pas nikkel.}
\end{claim}

    Actually, we will see that \cref{claim:goal} holds even if we restrict $\sU$ to be a forest. Notice that this claim will complete the proof of \cref{th:np-hard}.

    \myparagraph{Upper bound on $\cardV*{\sU}$.} The easy part of \cref{claim:goal} is to show that if $(X,Y,Z,T)$ has a matching $M \subseteq T$, then $\cardV*{\sU} \le \cardV*{F^X} + n$, even if we restrict $\sU$ to be a forest.
    
    From $M$, a small isometric-universal graph $\sU'$ for $\sF$ can be constructed as follows. For each triple $(x,y,z) \in M$, we construct the graph $U'(x,y,z)$ obtained from $F(x)$ by adding one vertex, say $v(y,z)$, that is connected to the vertex $p(y,z)$ of $F(x)$, the vertex between the centers $c(y)$ and $c(z)$ of the copies of $S(y)$ and $S(z)$. Note that the vertices $c(y),c(z),p(y,z)$ and $v(y,z)$ induced an isometric claw in $U'(x,y,z)$. See \cref{fig:Uxyz}.
    
    By construction of $F(x)$, since $(x,y,z) \in T(x)$, both stars $S(y)$ and $S(z)$ appear as isometric subgraphs of $F(x)$. Moreover, the claw containing $v(y,z)$ proves that a copy of $F(y)$ and of $F(z)$ appear also as isometric subgraphs in $U'(x,y,z)$, as well as a copy of $F(x)$ by construction of $U'(x,y,z)$. %\Cyril{Peut-être mettre un dessin focusé sur la partie de $U'(x,y,z)$ avec $S(1)$ et $S(3)$ par exemple, $n=2$. J'attends quand même la fin de la preuve, pour être sûr que les exemples ne vont pas changer.} \Edgar{Oui, ou alors un exemple avec une taille générique, faut voir ce qui est le plus simple à comprendre}
    Since $M$ covers once each of the elements of $X\cup Y\cup Z$, the graph $\sU' = \bigcup_{(x,y,z)\in M} U'(x,y,z)$ is an isometric-universal graph for $\sF$.

    \begin{myfigure}
    \begin{tikzpicture}[scale=0.9]
    \def\sep{3}% écart entre deux centres
    \def\S{{"","z","y"}} % liste des stars
    \def\H{{0,.3,.9}} % liste des hauteurs

    \tikzset{every node/.style={vertex}}
    \def\GET#1{\pgfmathparse{#1}\pgfmathresult}

    % dessine les triangles
    \foreach \i in {1,2} {
        \pgfmathsetmacro\x{\sep*(\i-1)}
        \pgfmathparse{\H[\i]}
        \pgfmathsetmacro\h{\pgfmathresult}
        \pgfmathsetmacro\z{2.5-\h}
        \coordinate (c\i) at (\x,0);
        \node[draw=none,yshift=-1.1em] at (c\i) {$c(\GET{\S[\i]})$};
        \node[draw=none,yshift=1.7em] at ($(c\i)+(0,2)$) {$S(\GET{\S[\i]})$};
        \draw[rounded corners=4pt] (c\i) -- ++(120:2.5) -- ++(0.5,0) -- ++(-60:\h) -- ($(60:\z)+(c\i)$) -- (c\i);
    };

    % chemin en les étoiles
    \node[draw=none,yshift=1em] at ($(c1) + (0.5*\sep,0)$) {$p(\GET{\S[1]},\GET{\S[2]})$};
    \coordinate (u1) at ($(c1) + (0.5*\sep,0)$); % sommet du milieu
    \coordinate (u0) at ($(c1) - (0.3*\sep,0)$);
    \coordinate (v0) at ($(u0) - (0.3*\sep,0)$);
    \coordinate (u2) at ($(c2) + (0.3*\sep,0)$);
    \coordinate (v2) at ($(u2) + (0.3*\sep,0)$);
    \draw (u0) -- (c1);
    \draw (u2) -- (c2);
    \draw[red,line width=1pt] (c1) -- (u1) -- (c2); % claw
    \draw[dotted] (v0) -- (u0);
    \draw[dotted] (v2) -- (u2);
    \node at (c2) {};
    \node at (c1) {};

    % sommet v(y,z)
    \coordinate (v) at ($(u1) - (0,0.4*\sep)$);
    \node[draw=none,yshift=-1.1em] at (v) {$v(y,z)$};
    \draw[red,line width=1pt] (u1) -- (v); % claw
    \node at (v) {};
    \node at (u1) {};

    \end{tikzpicture}
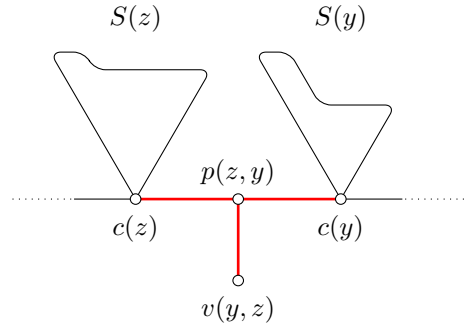
    \caption{The part of the component $U'(x,y,z)$ of $\sU'$ containing $F(y)$ and $F(z)$ as isometric subgraphs. Both subgraphs share their claw (in red).}
    \label{fig:Uxyz}
    \end{myfigure}

    We check that the number of vertices of $\sU'$ is:
    \begin{eqnarray*}
        \cardV{\sU'} &=& \sum_{(x,y,z)\in M} \cardV{U'(x,y,z)} ~=~ \sum_{(x,y,z)\in M} \pare{ \cardV{F(x)} + 1} \\
        &=& \pare{ \sum_{(x,y,z)\in M} \cardV{F(x)}} + |M|\\
        &=& \pare{ \sum_{x\in X} \cardV{F(x)} } + n ~=~ \cardV{F^X} + n\\
    \end{eqnarray*}
    Since $\cardV*{\sU} \le \cardV*{\sU'}$, we have $\cardV*{\sU} \le \cardV*{F^X} + n$, as required for \cref{claim:goal}. We observe that $\sU'$ is a forest, so the upper bound holds even if we restrict $\sU$ to be a forest.
    
    \myparagraph{Lower bound on $\cardV*{\sU}$.} The difficult part of \cref{claim:goal} is to show that if $(X,Y,Z,T)$ has no solutions, then $\cardV{\sU} > \cardV{F^X} + n$.

    We start by noticing that \cref{claim:two_cc}, in the proof of \cref{th:2graphs}, trivially extends to more than two graphs as follows:
    
    \begin{claim}\label{claim:three_cc}
        Each component of any isometric-universal graph for $\sF$ contains at most one tree of every forest of $\sF$.
    \end{claim}
    
    This is because any embedding that would pack two trees of a same forest of $\sF$ into the same connected component of the universal graph would be not isometric. 
    %\Cyril{Petit détail. Je pensais qu'il était clair que, de manière générale, tout GUI minimum pour trois graphes à $n$ CC devait avoir exactement $n$ CC (sinon, on en fusionne deux et c'est plus petit). Mais c'est faux! On pourrait avoir $n=2$ et pourtant avoir $n+1=3$ CC dans $\sU$: chaque CC du GUI minimum ayant exactement $2$ CC pris dans deux graphes différents à chaque fois. Car le nombre de paires de graphes pour $3$ graphes de $2$ CC est $\binom{3}{2} = 3 = n+1$. Du coup impossible de fusionner les CC par un seul sommet ! C'est pas général. On est obligé de le démontrer spécifiquement pour notre $\sF$. } \Edgar{Ah pétard t'as raison !} \Cyril{Ton exemple clique-chemin-étoile, chacun de taille $n$, est d'ailleurs très bon. Cela fait un GUI de $3n$ sommets avec 3 composantes à $n$ sommets et pourtant trois graphes de deux composantes connexes chacun de $n$ sommets :((~}

    By way of contradiction, we assume that $\cardV*{\sU} \le \cardV*{F^X} + n$, and consider any isometric embedding $\varphi: \sF \to \sU$. For convenience, for each $w \in X\cup Y\cup Z$, we denote by $\hF(w)$, the isometric copy of $F(x)$ in $\sU$ w.r.t. $\varphi$. We have:

    \begin{claim}\label{claim:number_cc}
        $\sU$ has exactly $n$ connected components, even if $\sU$ is a forest.
    \end{claim}
    
    \begin{proof}
        By \cref{claim:three_cc}, $\sU$ must have at least $n$ connected components. If $\sU$ has more than $n$ components, then one of them, say $U$, does not contain any tree of $F^X$, that is $U$ and $\hF(x)$ are disjoint subgraphs of $\sU$. By minimality of $\sU$, component $U$ must contain at least one tree $F(w)$ for some $w \in Y\cup Z$ (otherwise we could remove $U$ from $\sU$ and save vertices). So, $\varphi(F^X)$ and $\hF(w)$ are disjoint subgraphs of $\sU$. It follows that $\cardV*{\sU} \ge \cardV*{F^X} + \cardV*{F(w)}$. This is a contradiction with $\cardV*{\sU} \le \cardV*{F^X} + n$, since by construction $\cardV*{F(w)} > \cardV*{S(w)} > 3n$ (see \cref{claim:S(w)}). And, this holds even if $\sU$ is a forest.
    \end{proof}
    
    For every $w \in X\cup Y\cup Z$, denote by $U(w)$ the unique connected component of $\sU$ containing $\hF(w)$. For a component of $\sU$ containing $\hF(x), \hF(y)$ and $\hF(z)$, we have by definition $U(x) = U(y) = U(z)$. According to \cref{claim:three_cc} and \cref{claim:number_cc}, there are exactly $n$ triples $(x,y,z)$ such that $U(x) = U(y) = U(z)$. Such triples have to form a matching of $X\times Y\times Z$. In particular, since $(X,Y,Z,T)$ has no matching, we must have a some triple $(x_0,y_0,z_0) \notin T$ such that $U(x_0) = U(y_0) = U(z_0)$.

    Our goal is to show that:
    \begin{enumerate}%[noitemsep]
        \item\label{it1} $\cardV*{U(x)} \ge \cardV*{F(x)} + 1$, for all $x\in X$; and
        \item\label{it2} $\cardV*{U(x_0)} > \cardV*{F(x_0)} + 1$.
    \end{enumerate}
    Because $\cardV*{\sU} = \sum_{x\in X} \cardV*{U(x)}$ and $\cardV*{F^X} = \sum_{x\in X} \cardV*{F(x)}$, \cref{it1} and \cref{it2} are sufficient to prove that $\cardV*{\sU} > \cardV*{F^X} + n$, the desired lower bound for \cref{claim:goal}. These lower bounds hold independently whether $\sU$ is restricted to be a forest or not.

    To prove \cref{it1}, let us show that:
    
    \begin{claim}\label{claim:cardUx+Uw}
        For every $x \in X$ and $w\in Y\cup Z$ such that $U(x) = U(w)$, we have: $\cardV*{\hF(x) \cup \hF(w)} \ge \cardV*{F(x)} + 1$.%, even if $\sU$ is a forest.
    \end{claim}

    \begin{proof}
        Observe that $F(w)$ is not a subgraph of $F(x)$, because $F(w)$ has two adjacent vertices each of degree three or more, the center of $S(w)$ and its claw. This configuration does not appear in $F(x)$. Thus, there must exists an edge $u-v$ of $\hF(w)$ that is not in $\hF(x)$, both subgraphs being connected. So, if $\cardV*{\hF(x) \cup \hF(w)} = \cardV*{\hF(x)}$, i.e., $\hF(x) \cup \hF(w)$ has the same number of vertices as $\hF(x)$ has, then $u$ and $v$ are both in $\hF(x) \cup \hF(w)$. This is a contradiction with the fact $\hF(x)$ must be isometric in $U(x) \supseteq \hF(x) \cup \hF(w)$: the distance between $u$ and $v$ in $U(x)$ and $\hF(x)$ differs. It follows that $\hF(x) \cup \hF(w)$ has more vertices as $\hF(x)$, i.e., $\cardV*{\hF(x) \cup \hF(w)} \ge \cardV*{\hF(x)} + 1 = \cardV*{F(x)} + 1$. % We note that this independently holds whether $\sU$ is a tree of not. 
    \end{proof}

    Since $U(x)$ has to contain $\hF(x)$ and $\hF(w)$ for some $w \in Y\cup Z$, we have proved \cref{it1}. Now, let us prove \cref{it2}, that is:
    
    \begin{claim}\label{claim:cardUx0}
        $\cardV*{U(x_0)} > \cardV*{F(x_0)} + 1$. %, even if $\sU$ is a forest.
    \end{claim}

    \begin{proof}
        Recall that $(x_0,y_0,z_0) \notin T$ and $U(x_0) = U(y_0) = U(z_0)$, i.e., the component $U(x_0)$ contains each of the subgraphs $\hF(x_0), \hF(y_0),\hF(z_0)$. There are several cases (and sub-cases) to analyze depending on the contents of $T(x_0)$.

        \textbf{\boldmath{Case 1. There is $w\in \set{y_0,z_0}$ that does not appear in $T(x_0)$.}}
        
        It means that $F(x_0)$ does not contain $S(w)$. For convenience, let $U = \hF(x_0) \cup \hF(w)$. Because $U(x_0) \supseteq U$, it suffices to show that $\cardV*{U} > \cardV*{F(x_0)} + 1$, a stronger version of \cref{claim:cardUx+Uw}.
        
        Let $c$ be the vertex of $U$ that corresponds to the center of $\hF(w)$, i.e., the center of the copy of $S(w)$ contained in $\hF(w)$. By \cref{claim:S(w)}, $c$ has degree $(3w+3)+1$ in $\hF(w)$ including its claw. In $U$, $c$ must have a degree at least $3w+4$, otherwise $\hF(w)$ would be not an isometric subgraph of $U$. Moreover, $3w+4$ of its neighbors in $U$ must induce an independent set ($\hF(w)$ being a tree), otherwise $\hF(w)$ would be not an isometric subgraph of $U$. 
        
        \textbf{Sub-case 1.1:} $c$ belongs to $\hF(x_0)$.
        
        If $c$ has degree two or less in $\hF(x_0)$, then at least $(3w+4)-2 \ge 2$ outgoing edges of $c$ in $U$ do not belong to $\hF(x_0)$. Each of these edges has only one endpoint in $\hF(x_0)$, otherwise it would create a shortcut in $U$ and so $\hF(x_0)$ would be not an isometric subgraph of $U$. This implies that $U$ contains at least two vertices that are not in $\hF(x_0)$. In other words, $\cardV*{U} > \cardV*{F(x_0)} + 1$ as required.
        
        If $c$ has degree more than two in $\hF(x_0)$, it must correspond to a center of some star $S(w')$ of $F(x_0)$, with $w'\neq w$ by assumption on $w$. In $\hF(x_0)$, $c$ has degree $3w'+4$ or $3w'+5$.
        
        If $w'<w$, i.e., $w' \le w-1$, then $3w'+5 \le 3(w-1)+5 = 3w+2$. So, two out of the $3w+4$ outgoing edges of $c$ in $U$ are not in $\hF(x_0)$. As see previously, this implies that two extra vertices out of $\hF(x_0)$ are required, and so $\cardV*{U} > \cardV*{F(x_0)} + 1$ as claimed.
        
        If $w'>w$, then $S(w)$ contains three longer branches than $S(w')$ (by one edge or more from \cref{claim:S(w)}). At least one of these long branches exceeds all the branches of $S(w')$ (two long branches of $S(w)$ can be embedded on the main path of $\hF(x_0)$). We check that that the claw of $\hF(w)$ creates also at least one edge out of $\hF(x_0)$. Altogether, there $c$ has at least two outgoing edges of $\hF(x_0)$ providing two vertices out of $\hF(x_0)$, and implying $\cardV*{U} > \cardV*{F(x_0)} + 1$ as claimed.

        \textbf{Sub-case 1.2:} $c$ does not belong to $\hF(x_0)$.
        
        We can assume that $\cardV*{U} \le \cardV*{F(x_0)} + 1$, since otherwise we are done. In this case, since $c$ is not in $\hF(x_0)$, then all the neighbors of $c$ in $U$ must be in $\hF(x_0)$. Moreover, this neighbors must be included in the neighbors of the center $c'$ of some copy of $S(w')$ in $\hF(x_0)$. If not, we can check that some distances in $\hF(x_0)$ or in $\hF(w)$ would be shortened (remember that $3w+4$ neighbors of $c$ must be an independent set in $U$ and that different stars of $\hF(x_0)$, apart their centers, are at distance greater than two). One can reproduce the same analysis as in \cref{claim:cardUx+Uw} where $c$ plays the role of the center of $S(w')$ and $w'>w$. And thus, $\cardV*{U} > \cardV*{F(x_0)} + 1$ as required.

        This completes the first case, observing that the proof holds even if $U$ is a tree.
        
        \textbf{\boldmath{Case 2. $y_0$ and $z_0$ appear in $T(x_0)$.}}
        
        It means that the tree $F(x_0)$ contains $S(y_0)$ and $S(z_0)$. However, since no triple of $T(x_0)$ contains both $y_0$ and $z_0$, then $S(y_0)$ and $S(z_0)$ do not appear consecutively in the main path of $F(x_0)$, and their distance is at least five. (Remember that in the construction of $F(x)$, stars appear along the subdivided path in the specific order $S(y_1) -\circ- S(z_1) -\circ-\circ- S(z_2) -\circ- S(y_2) \cdots$, cf. \cref{fig:F(x)}.) 

    Let $c$ be the vertex of $U$ that corresponds to the center of $\hF(y_0)$ and $c'$ be the vertex of $U$ that corresponds to the center of $\hF(z_0)$. If $c \in \hF(x_0)$, then by the proof of Sub-case 1.1, $c$ must match the center of $S(y_0)$ in $F(x_0)$. 
    If $c \notin \hF(x_0)$, then its neighborhood must be contained in that of the center of $S(y_0)$ in $F(x_0)$, as in the proof of Sub-case 1.2.
    Note that in both cases, the neighborhood of $c$ must be contained in that of the center of $S(y_0)$ in $F(x_0)$.

    Similarly, the neighborhood of $c'$ must be contained in that of the center of $S(y_0)$ in $F(x_0)$.
    
    Now $c$ and $c'$ are at distance at least $5$, and in $\hF(x_0)$, $c$ and $c'$ each has an adjacent vertex of degree $3$. These vertices of degree $3$ must each have a neighbor that is not in $\hF(x_0)$ (otherwise the distances of $F(x_0)$ are not preserved). Therefore there is a vertex at distance either $0$ or $2$ from $c$ that is not in $\hF(x_0)$, and a vertex at distance either $0$ or $2$ from $c'$ that is not in $\hF(x_0)$. These two vertices are not the same, otherwise the distance between $c$ and $c'$ would not be correct. Therefore $\cardV*{U} > \cardV*{F(x_0)} + 1$ as required.
    
    As said previously, by summing up the number of vertices of all the components of $\sU$ and of $F^X$, and by applying \cref{claim:cardUx+Uw} and \cref{claim:cardUx0}, we have showed that $\cardV*{\sU} > \cardV*{F^X} + n$. This proves the desired lower bound part for \cref{claim:goal}.
    \end{proof}
    
This completes the proof of \cref{th:np-hard}.
\end{proof}

\section{Conclusion and Open Question}
\label{sec:conclusion}

We have introduced the notion of $k$-isometric-universal graph, which captures several well-studied variants of universal graph, by taking different values for $k$, with $k=0$ (classical universal graphs), $k=1$ (induced-universal graphs),
or $k=\infty$ (isometric-universal graphs). We have also considered limited families of graphs (by taking two graphs, two trees or forests, all $n$-vertex graphs, ...).

We have mainly proved that, for $k=\infty$, the problem of minimizing the number of vertices of a minimum $k$-isometric-universal graph is polynomial for two forests and NP-hard for three forests. The variants for $k \in\set{0,1}$ are known to be NP-hard.

Along the way, we have proved that any minimum and minimal isometric-universal graph for two trees has to be a tree (\cref{th:2trees}). We propose two open questions, one about the structure of minimum isometric-universal graphs, and one about the complexity for two graphs. More precisely,

\begin{myquote}
    What is the arboricity of a minimum and minimal isometric-universal graph of $t>2$ trees? Is it $t-1$?
\end{myquote}

\noindent
\cref{th:2trees} shows that this is indeed $t-1$ for $t = 2$. Note that \cref{th:3trees} shows that for $t = 3$ trees, arboricity~1 is not possible.

\begin{myquote}
    What is the complexity of finding a minimum isometric-universal graph for two outerplanar graphs? Or two graphs of bounded treewidth?
\end{myquote}

\bibliographystyle{my_alpha_doi}
\def\MYMARGIN{15mm}
% Permet d'avoir des un bibliographystyle{my_alpha_doi}
% avec le style LIPIcs.
%
% C'est la copie de la redéfinition de la commande
% \thebibliography par lipics-v2021.cls où j'ai remplacé
% un \leftmargin8.5mm par \leftmargin\MYMARGIN que l'on
% peut fixer à 18.5mm ou moins.
%
% Cyril - 27/05/2025

\makeatletter
\ifdefined\MYMARGIN\else\def\MYMARGIN{18.5mm}\fi
\renewenvironment{thebibliography}[1]
  {\if@noskipsec \leavevmode \fi
   \par
   \@tempskipa-3.5ex \@plus -1ex \@minus -.2ex\relax
   \@afterindenttrue
   \@tempskipa -\@tempskipa \@afterindentfalse
   \if@nobreak
     \everypar{}%
   \else
     \addpenalty\@secpenalty\addvspace\@tempskipa
   \fi
   \noindent
   \rlap{\color{lipicsLineGray}\vrule\@width\textwidth\@height1\p@}%
   \hspace*{7mm}\fboxsep1.5mm\colorbox[rgb]{1,1,1}{\raisebox{-0.4ex}{%
     \normalsize\sffamily\bfseries\refname}}%
   \@xsect{1ex \@plus.2ex}%
   \list{\@biblabel{\@arabic\c@enumiv}}%
        {\leftmargin\MYMARGIN
         \labelsep\leftmargin
         \settowidth\labelwidth{\@biblabel{#1}}%
         \advance\labelsep-\labelwidth
         \usecounter{enumiv}%
         \let\p@enumiv\@empty
         \renewcommand\theenumiv{\@arabic\c@enumiv}}%
   \fontsize{9}{12}\selectfont
   \sloppy
   \clubpenalty4000
   \@clubpenalty \clubpenalty
   \widowpenalty4000%
   \sfcode`\.\@m\protected@write\@auxout{}{\string\gdef\string\@pageNumberStartBibliography{\thepage}}}
\makeatother

\bibliography{biblio}

%%%%%%%%%%%%%%%%%%%%%%%%%%%%%%%%%%%%%%%%%%%%%%%%%%%%%%%%%%%%%%%%%%%%%% 

\end{document}